\documentclass[10.5pt, letter]{article}

\usepackage[margin=0.8in]{geometry}
\usepackage{libertine}
\usepackage{libertinust1math}
\usepackage[T1]{fontenc}
\usepackage[norelsize,ruled,vlined,linesnumbered]{algorithm2e}
\usepackage{algorithmic}
\usepackage{amsthm}
\usepackage{amsmath,amssymb,amsfonts}
\usepackage{graphicx}
\usepackage{subfigure}
\usepackage{comment}
\usepackage{hyperref}
\usepackage{xcolor}
\usepackage[numbers]{natbib}
\usepackage{booktabs}
\usepackage{cleveref}

\newtheoremstyle{mystyle}
    {2.0mm}
    {2.0mm}
    {\it}
    {0.0mm}
    {\sc}
    {.}
    { }
    {}
\theoremstyle{mystyle}
\newtheorem{definition}{Definition}
\newtheorem{lemma}{Lemma}
\newtheorem{theorem}{Theorem}

\newtheorem{proposition}{Proposition}

\newcommand{\vs}{\vspace{1.5mm}}
\SetKwComment{Comment}{$\triangleright$\ }{}

\SetKwComment{Comment}{$\triangleright$\ }{}

\title{Grant: A Framework for Approximate Nearest Neighbor Search with Multi-Attribute Range Filters}
\author{
    Daichi Amagata\\[2mm]
    The University of Osaka \\
    \textit{amagata.daichi@ist.osaka-u.ac.jp}
}
\date{}

\begin{document}

\maketitle

\begin{abstract}
Recently, academia and industry have considered the problem of approximate nearest neighbor search (ANNS) with filters.
In this setting, each object consists of a high-dimensional vector and attribute values.
Given a query vector, filters for attribute values, and $k$, this problem retrieves $k$ vectors approximately nearest to the query vector among a set of objects passing the filters.
This paper considers range filters, i.e., users can specify a range constraint for each attribute.
Most existing works do not consider this setting, and they assume (i) only a single attribute or (ii) matching filters that require the same attribute values or categories.
Existing techniques for these assumptions are not available for our setting or are trivially not efficient.
Although standard solutions, such as pre-filter and post-filter, can handle our problem, they are also inefficient.
Some works tackle the same problem as ours, but their techniques necessitate historical query workloads, which significantly limit practical use cases.
To remove these limitations, this work proposes Grant, a novel framework that solves this problem efficiently while accepting arbitrary range filters and ANNS data structures.
Grant can guarantee a search time sub-linear to the number of objects, which is not held by existing techniques.
We conduct extensive experiments, and their results demonstrate that Grant outperforms existing techniques.
\end{abstract}

\section{Introduction}  \label{sec:introduction}
Modern search engines increasingly rely on machine learning models that embed many types of objects, such as images and texts, into high-dimensional vectors.
Due to this fact, the problem of similarity search in a high-dimensional space becomes more important in many fields, such as recommender systems \cite{amagata2021reverse,amagata2023reverse}, information retrieval \cite{amagata2025mine,amagata2025approximate,aoyama2023simpler,hirata2023categorical,hirata2022solving}, databases \cite{nakama2021approximate,uemura2025ivf++,yamashita2025should,uemura2024efficient}, artificial intelligence \cite{kido2025fast}, and machine learning \cite{amagata2021fast,amagata2022fast}.
For example, computing the exact nearest neighbor on large data with high dimensionality is computationally expensive, and many applications, such as RAG (Retrieval-Augmented Generation), allow approximate results \cite{li2019approximate}.
Therefore, approximate nearest neighbor search (ANNS) has been studied extensively \cite{andoni2017nearest}.

Nowadays, to enable making use of attribute information associated with objects, ANNS with attribute filters has been considered both in academia \cite{wu2022hqann,gupta2023caps,wang2023efficient,zuo2024serf,xu2024irangegraph,patel2024acorn} and industry (e.g., Alibaba \cite{wei2020analyticdb}, Apple \cite{mohoney2023high}, and Microsoft \cite{zhang2023vbase,gollapudi2023filtered}).
As specified in these industrial works, the problem of ANNS with attribute filters is also important in many real-world Web applications.
A typical example of this problem is that, in e-commerce scenarios, each user finds the products most similar to a given query vector among a set of objects satisfying attribute filters (e.g., price, rating, and production date).
Similarly, RAG applications receive benefits from ANNS with attribute filters.
For example, when a user wants to know health risk, he/she may write a prompt like ``investigate what symptoms tend to appear for people aged from 40 to 50 and with blood pressure at least 140.''
This script is transformed into an ANNS with range filters on age and blood pressure attributes, and the large language model obtains the answer from the retriever.

\subsection{Limitations of Prior Work}
The existing works on the problem of ANNS with attribute filters have the following issues.

\vs
\noindent
(1) \textit{Single attribute assumption.}
Many existing works \cite{wu2022hqann,zuo2024serf,xu2024irangegraph,patel2024acorn,liang2025unify,xie2025beyond,wang2025wow,ye2025compass} assume that each object has only a single attribute, and their techniques are available only in this case.
(Although some works \cite{xu2024irangegraph,zuo2024serf,engels2024approximate,liang2025unify,ye2025compass} consider range filters, they cannot deal with multiple attributes, or they are reduced to the standard technique, Post-filter, which is introduced in \Cref{sec:preliminary}.)
This setting significantly limits applications, and objects generally have multiple attributes, as seen in the above examples.

\vs
\noindent
(2) \textit{Dealing with only matching filters.}
Most existing works consider attribute ``matching'' as a filter \cite{jiang2025digra,wu2022hqann,gupta2023caps,gollapudi2023filtered,wang2023efficient,patel2024acorn,cai2024navigating}, i.e., a filter evaluates whether a given object has the same attribute value/category as that of this filter.
As the above application examples show, users typically specify ``range'' filters, e.g., they want items whose prices and production dates fall within certain ranges.
Obviously, existing techniques for matching filters are inefficient for range filters on multiple attributes, as they require running for each filter in multiple continuous ranges, which necessitates an infinite number of iterations. 

\vs
\noindent
(3) \textit{Always requiring past queries.}
Although some works \cite{mohoney2023high,li2025sieve} consider range filters on multiple attributes, they necessitate past queries.
They build ANNS indices based on past queries, so they do not function without past queries. 
This strong assumption significantly limits practical use cases.
A better and practical property is that an ANNS algorithm works well without past query workloads and its efficiency can be accelerated with past queries.
Note that no work devised optimal algorithms for ANNS with range filters on multiple attributes, without past query workloads.

\vs
\noindent
(4) \textit{Vector data management systems are general but inefficient.}
Some vector data management systems, e.g., \cite{wei2020analyticdb,wang2021milvus,sehgal2025navix}, support multiple attributes and range filters.
They, however, do not optimize ANNS algorithms for this setting, and their approaches are essentially Pre-filter or Post-filter.

\subsection{Contribution}
This paper tackles the problem of ANNS with multi-attribute range filters, a more general and challenging problem than those with a single attribute and matching filters.
An efficient and generally available solution (i) must be a workload-free structure and (ii) has an \textit{option} of past query usage for efficiency acceleration.
No existing works satisfy efficiency and generality at the same time for this problem.

This paper proposes \textit{Grant} (\underline{Gra}ph i\underline{n} \underline{t}ree index), an efficient framework for this problem, and it incorporates graph indices into a tree index.
It has three main advantages:
(i) It handles multiple attributes and arbitrary ranges, optimized for practical scenarios.
(ii) It can employ arbitrary state-of-the-art multi-dimensional tree index and proximity graph.
(iii) It functions well without any past query workloads, and its efficiency can be enhanced with past query vectors.
(iv) It is fast and accurate in practice.
To summarize, our main contributions are as follows.
\begin{itemize}
    \setlength{\leftskip}{-4.0mm}
    \item   We propose Grant, a new and general framework for the problem of ANNS with multi-attribute range filters, and it is built without any past query workloads.
    \item   We present an efficient range-filtered ANNS algorithm that exploits the Grant structure.
            If past query vectors are available, the efficiency of the proposed algorithm can be accelerated.
    \item   When we use a kd-tree \cite{bentley1975multidimensional} as a base structure of Grant, Grant can have a theoretical guarantee that the ANNS time is sub-linear to $n$, where $n$ is the number of objects.
            This theoretical result is not achieved by existing algorithms for ANNS with range filters on multiple attributes.
    \item   We conduct extensive experiments, and their results demonstrate that Grant outperforms state-of-the-art baselines in most tests.
\end{itemize}

\section{Preliminary}   \label{sec:preliminary}

\subsection{Problem Definition}
Let $X$ be a set of $n$ objects.
Each object $x_i \in X$ has a $d$-dimensional vector $\mathbf{v}_i$ and $m$ attribute values $\mathbf{a}_i$, i.e., $x_i = (\mathbf{v}_i, \mathbf{a}_i)$ where $d$ is large and $m \geq 2$.
Notice that the $m$ attribute values can be represented as an $m$-dimensional vector, so we use $a_i[j]$ to denote the $j$-th attribute value of $\mathbf{a}_i$.
Let $dist(\cdot,\cdot)$ be the distance between two vectors.
The problem of $k$ nearest neighbor search with multi-attribute range filters is defined as follows:

\begin{definition}[\textsc{$k$ nearest neighbor search with multi-attribute range filters}]
\textit{Given an object set $X$, a query vector $\mathbf{q}$, query range filters $F = \{[l_1, u_1], ..., [l_m, u_m]\}$, and a result size $k$, this problem retrieves a set $S^*(\mathbf{q},F)$ of $k$ objects such that (i) $S^*(\mathbf{q},F) \subseteq X(F)$ and (ii) $\forall x \in S^{*}(\mathbf{q},F)$ and $\forall x' \in X(F)\backslash S^{*}(\mathbf{q},F)$, $dist(\mathbf{v},\mathbf{q}) \leq dist(\mathbf{v'},\mathbf{q})$, where}
\begin{equation*}
    X(F) = \{x \,|\, x \in X, \forall j \in [1,m], l_j \leq a[j] \leq u_j\}.
\end{equation*}
\textit{Ties are broken arbitrarily.
If $|X(F)| < k$, we return $X(F)$.}
\end{definition}

\noindent
This paper considers an approximation version of the above problem.
We search for $S(\mathbf{q},F)$ such that $S(\mathbf{q},F) \subseteq X(F)$, $|S(\mathbf{q},F)| = k$, but we may have $S(\mathbf{q},F) \neq S^{*}(\mathbf{q},F)$.
To measure the accuracy of $S$, we use recall@$k$, a criterion commonly used in ANNS studies, and
\begin{equation*}
    recall@k = \frac{|S(\mathbf{q},F) \cap S^{*}(\mathbf{q},F)|}{k}.
\end{equation*}

This work assumes an in-memory setting, and disk-based approaches \cite{subramanya2019diskann,chen2021spann,zhang2025highly} are not within the scope of this paper.
The main notations in this paper are summarized in \Cref{tab:notation}.
When we analyze algorithms, we assume that $d = O(1)$ because our algorithms linearly scale to $d$.

\begin{table}[!t]
    \centering
    \caption{Summary of notations}
    \label{tab:notation}
    \begin{tabular}{cl} \toprule
        Notation                & Meaning       \\ \midrule
        $x$                     & Object        \\ 
        $\mathbf{v}$            & $d$-dimensional vector of $x$                 \\
        $\mathbf{q}$            & $d$-dimensional query vector                  \\
        $\mathbf{a}$            & $m$ attribute values ($m$-dim. vector) of $x$ \\
        $X$                     & Set of $n$ objects    \\
        $dist(\cdot,\cdot)$     & Distance between two $d$-dimensional vector   \\
        $F$                     & Query filters \\
        $k$                     & Search result size    \\
        $S^{*}(\mathbf{q},F)$   & Exact $k$NN result of $\mathbf{q}$ and $F$  \\
        \bottomrule
    \end{tabular}
\end{table}

\subsection{Graph Index for ANNS}
In a graph index, each vertex corresponds to a vector, and similar vectors are connected by edges.
Edge selection is dependent on a given graph structure:
The most representative graph indices, such as HNSW \cite{malkov2018efficient}, NSG \cite{fu2019fast}, and Vamana \cite{subramanya2019diskann}, follow \textit{relative neighbor graph} or RNG structure.
Consider three vectors $\mathbf{v}_1$, $\mathbf{v}_2$, and $\mathbf{v}_3$.
If $dist(\mathbf{v}_2,\mathbf{v}_3) < dist(\mathbf{v}_1,\mathbf{v}_3)$, the RNG does not create an edge between $\mathbf{v}_1$ and $\mathbf{v}_3$.

ANNS on a graph index is done by greedy search.
Given a query vector $\mathbf{q}$, the greedy search traverses the graph index from its entry vertex.
For each unvisited neighbor vertex $\mathbf{v}$, the greedy search computes $dist(\mathbf{v},\mathbf{q})$ and pushes $\mathbf{v}$ into a heap, so that we can next move to the nearest vertex to $\mathbf{q}$.
We iterate this traversal until the heap becomes empty (i.e., no closer vectors are found).


\subsection{Baseline Algorithms}
The most straightforward algorithm for our problem is a brute-force one that runs a linear scan of $X$ (and it computes $dist(\mathbf{v}_i,\mathbf{q})$ only when $\mathbf{a}_i$ satisfies $F$).
Clearly, this algorithm incurs inefficient $O(n)$ distance computations, so existing works devised the following techniques.

\vs
\noindent
\textbf{\textsf{Pre-filter}} first obtains $X(F)$ and then runs a linear scan of $X(F)$.
Notice that $X(F)$ can be efficiently obtained with a data structure based on attribute values.
Although this approach can return the exact answer, its search time is slow.
Because $X(F)$ is query-dependent, Pre-filter cannot use any data structure for ANNS.

\vs
\noindent
\textbf{\textsf{Post-filter}} uses a pre-built ANNS index (e.g., graph index \cite{wang2021comprehensive}) to accelerate query efficiency.
Given a query vector $\mathbf{q}$ and filters $F$, it runs a ``vanilla'' ANNS algorithm on the index and checks the search result set.
If the search result set does not contain $k$ objects that satisfy $F$, it runs another ANNS by enlarging $k$.
This is repeated until $k$ objects that satisfy $F$ are returned.
Because the search strategy of Post-filter does not care about attribute values, its search space can be large and requires many search iterations, leading to suboptimal performance.

\vs
\noindent
\textbf{\textsf{ACORN}} \cite{patel2024acorn}, a state-of-the-art filter-agnostic index, is available for our problem.
Its structure is a variant of HNSW, and ACORN traverses only neighbors that pass a given query filter.
This approach may fail to reach the exact $k$NN objects if no neighbors have the query filter.
It is not guaranteed that neighbors have similar attribute values.
Therefore, ACORN loses paths to $\mathbf{q}$ when no neighbors pass $F$, resulting in low recall.
Our experimental results confirm this observation.

\section{Related Work}  \label{sec:related-work}
Since the problem of ANNS in high-dimensional space has been used in many applications, its efficient algorithms have been extensively devised.
Although some approaches exist for ANNS, it is widely known that the graph index shows the best time-recall trade-off among them \cite{li2019approximate}.
To obtain a better time-recall trade-off, many graph indices were designed so far, e.g., \cite{malkov2014approximate,harwood2016fanng,fu2016efanna,fu2019fast,fu2021high,arai2021lgtm,peng2023efficient}.
HNSW \cite{malkov2018efficient} and Vamana \cite{subramanya2019diskann} are famous graph indices.
They are usually used as an underlying data structure in vector data management systems \cite{pan2024survey} and have theoretical analysis \cite{prokhorenkova2020graph,indyk2023worst}.
Grant allows an arbitrary graph index to maintain the high-dimensional vectors in $X$.

\vs
\noindent
\textbf{ANNS with Matching Filter.}
Many works consider this setting \cite{wu2022hqann,gollapudi2023filtered,wang2023efficient,patel2024acorn,cai2024navigating,simhadri2024results}.
For example, FilteredVamana \cite{gollapudi2023filtered} extends Vamana by adding a label for each edge.
During ANNS, only edges with the label matching a given filter are traversed.
StitchedVamana is another variant of Vamana, and it is a Vamana of the objects having the same label (or attribute value).
That is, for each distinct label, a StitchedVamana is built.
NHQ \cite{wang2023efficient} defines a customized distance function that considers vector and label similarities.
A proximity graph is built based on this customized distance function.
They do not support multi-attribute cases and/or range filters.

\vs
\noindent
\textbf{ANNS with Range Filter.}
Recent studies \cite{zuo2024serf,xu2024irangegraph} considered how to support a range filter in HNSW.
SeRF \cite{zuo2024serf} requires $O(n^2)$ space, and iRangeGraph reduced this space to $O(n\log n)$.
Another work \cite{engels2024approximate} developed a different data structure based on a segment tree, and it also needs $O(n\log n)$ space.
Literature \cite{liang2025unify} further reduces the space complexity to $O(n)$.
These works focus on a single attribute case, and \cite{xu2024irangegraph} suggests that their approaches are reduced to Post-filter when query filters specify multiple attributes.
WoW \cite{wang2025wow} achieves an updatable index for ANNS with a single range filter.

HQI \cite{mohoney2023high} and SIEVE \cite{li2025sieve} can deal with multi-attribute range filters, but they require historical query filters to build the structures.
That is, without this assumption, their structures cannot be built.
(Without past filters, SIEVE is reduced to a standard algorithm, e.g., Post-filter on a graph index.)
Furthermore, HQI focuses on \textit{batch query processing}, and this setting is different from our objective, i.e., processing each ANN query efficiently while providing a high recall result.
Unlike these works, the Grant structure can be built without any past information.

\vs
\noindent
\textbf{Vector Data Management Systems.}
Milvus \cite{wang2021milvus}, NaviX \cite{sehgal2025navix}, and ADBV \cite{wei2020analyticdb} support multi-attribute range filters.
Milvus can select an ANNS index, e.g., graph index and IVF (inverted file index), whereas ADBV employs IVF.
ADBV proposes two approaches: one is Pre-filter, and the other is the verification of a subset of $X$.
Specifically, when IVF is used, only some posting lists similar to a given query vector are verified.
This verification is the same as the brute-force algorithm introduced in \Cref{sec:preliminary}.
NaviX also employs Pre-filter but runs an approximate verification.
Milvus also employs Pre-filter and considers Post-filter.

\section{Grant} \label{sec:grant}
\textbf{Main idea.}
Indexing low-dimensional data is a classic problem, and many famous tree-based indices, such as kd-tree \cite{bentley1975multidimensional} and R-tree \cite{guttman1984r}, are known.
These data structures support efficient orthogonal (or axis-parallel) range queries.
Because the attribute values held by the objects in $X$ are regarded as low-dimensional vectors, we can build a tree index based on the attribute values.
Then, for example, $X(F)$ can be efficiently obtained with a range query of $F$.
This approach avoids accessing many vectors that do not satisfy $F$.
However, simply applying this idea results in Pre-filter, so we suffer from a high computational cost when $|X(F)|$ is large.
(If $n$ is large, e.g., tens of millions, $|X(F)| = 0.01 \cdot n$ is still large.)

On the other hand, assume that $X$ is indexed by a graph index (i.e., based on the $d$-dimensional vectors).
Although this approach yields a fast search time for ANNS without filters, our approach is limited to Post-filter when filters are specified.
Post-filter may need to access many objects (vectors) that do not satisfy $F$.

Now we have a challenging question:
\textit{Given a query vector $\mathbf{q}$ and query filters $F$, how can we access only objects that satisfy $F$ and are similar to $\mathbf{q}$?}
The above observations of the tree and graph indices actually suggest an idea yielding a positive answer to this question.
\textit{If we receive benefits of the tree and graph indices (if we can run graph-based ANNS on the range query result of $F$), we overcome the above challenge.} 
We propose Grant, which achieves this idea.

\subsection{Grant Structure}
Existing indices do not have the ability to navigate objects that pass a given filter and have vectors similar to $\mathbf{q}$.
\textit{The novelty of Grant is this navigation ability.}
Grant is a tree structure in which leaf nodes have a capacity of $O(n/c)$, where $c$ is a constant and hyper-parameter.
This tree structure is constructed based on the attribute values.
Each object in $X$ is maintained by a leaf node of the Grant, and each leaf node maintains the $O(n/c)$ objects by using a graph index $\mathcal{G}$ based on their $d$-dimensional vectors.
Note that the number of leaf nodes in Grant is $O(c) = O(1)$.

\Cref{fig:grant} illustrates an example of the Grant structure, assuming that each node can have at most two child nodes.
This combination of tree and graph structures is simple yet effective: The main merits of this structure are as follows.
(i) A set of objects having similar attribute values can be maintained by a graph index, i.e., attribute-based and vector-based navigation is simultaneously achieved.
(ii) No past information is required.
These properties are not held by existing techniques for our problem and are important to positively answer the challenging question mentioned earlier and yield a fast ANNS time described in \Cref{theorem:time}.

\begin{figure*}[!t]
    \centering
    \includegraphics[width=0.99\linewidth]{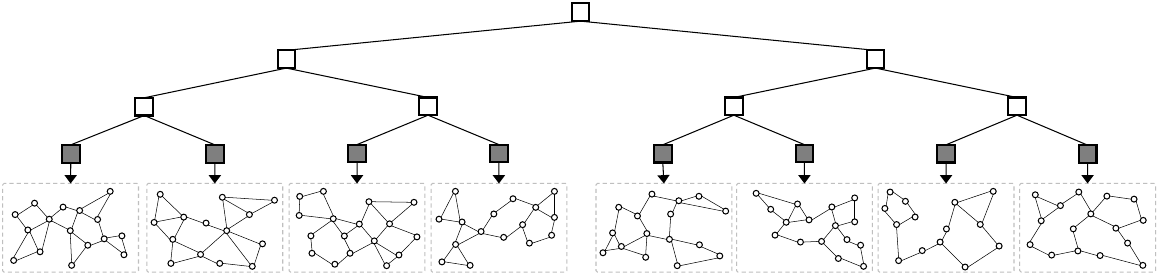}
    \caption{Grant structure. Grey nodes represent leaf nodes, each of which maintains a graph index.}
    \label{fig:grant}
\end{figure*}

\vs
\noindent
\textbf{Grant construction.}
This construction is done only once, and the Grant is general to any query vectors and range filters.
Let $\mathcal{T}$ be the base tree structure of Grant.
We use a top-down tree-building approach, and the Grant is built based on the construction algorithm of $\mathcal{T}$.
We assume that it is a balanced tree.

\Cref{algo:build} describes the constuction algorithm.
Building $\mathcal{T}$ is performed by \textsc{Make-Node}$(\cdot,c)$.
Given $X$, we make a root node of $\mathcal{T}$.
If $|X| \leq n/c$, it is a leaf node, and we build a graph index $\mathcal{G}$ based on the $d$-dimensional vectors in $X$.
($\mathcal{G}$ is an arbitrary graph index.)
Otherwise, we equally split $X$ based on the construction algorithm of $\mathcal{T}$ (i.e., based on attribute values) to make $\mathcal{T}$ balanced.
We then create its descendant nodes in a recursive manner.

\begin{algorithm}[!t]
    \caption{\textsc{Build}$(X,c)$}	\label{algo:build}
    \DontPrintSemicolon
    \KwIn {$X$, $c$}
    root node of $\mathcal{T} \gets$ \textsc{Make-Node}$(X,c)$\;
    \textbf{return} $\mathcal{T}$\;
    \SetKwProg{Fn}{Function}{:}{}
    \Fn{\textsc{Make-Node}$(Y,c)$}
    {
        $u \gets$ a new node of $\mathcal{T}$\;
        \If {$|Y| \leq n/c$}
        {
            \textsc{Build-Graph}$(Y)$\;
            \textbf{return} $u$
        }
        $Y_l$, $Y_r \gets$ \textsc{Split}$(Y)$\;    \label{algo:build:split}
        left child node of $u \gets$ \textsc{Make-Node}$(Y_l, c)$\;
        right child node of $u \gets$ \textsc{Make-Node}$(Y_r, c)$\;
        \textbf{return} $u$
    }
\end{algorithm}

\begin{proposition} \label{proposition:time}
Let $f(n)$ be the time of splitting a given set of $n$ objects.
Also, let $g(n)$ be the time of building a given proximity graph of the object set.
A Grant is built in $O(n(1 + f(n))) + g(n)$ time\footnote{As seen in \Cref{sec:experiment}, Grant is built quickly, e.g., within a few minutes on ten million objects.
Hence, as long as $X$ is not frequently updated, rebuilding Grant is practical for dynamic $X$ (e.g., batch updates).}.
\end{proposition}

\begin{proof}
Since $\mathcal{T}$ is a balanced tree, the height of $\mathcal{T}$ is at most $\log c = O(1)$ (recall that each leaf node can manage at most $n/c$ objects and $c = O(1)$).
In addition, the number of nodes of $\mathcal{T}$ is $\leq 2c-1=O(c) = O(1)$.
Hence, $\mathcal{T}$ can be built in $O(n(\log c + f(n))) = O(n(1 + f(n)))$ time.
The remaining cost we need to consider is that of building a graph index in each leaf node.
When we call \textsc{Build-Graph}$(Y)$ in \Cref{algo:build}, $|Y| \leq n/c$.
Because the number of leaf nodes is at most $c = O(1)$, the total cost of building the graph indices is $g(n)$.
It is now clear that this proposition holds.
\end{proof}

\begin{proposition}[Space complexity of a Grant]    \label{proposition:space}
Let $S(n)$ be the space of a given proximity graph, and a Grant needs $O(1) + S(n)$ space.
\end{proposition}

\begin{proof}
Recall the following three facts.
(i) Each object in $X$ is maintained in a leaf node of $\mathcal{T}$.
(ii) $\mathcal{T}$ has $O(c) = O(1)$ nodes, because the Grant has $c = O(1)$ leaf nodes.
(iii) Each graph has $S(n)$ space.
Given these, this proposition is true.
\end{proof}

\subsection{ANNS on Grant}  \label{sec:grant:query}
\Cref{algo:anns} describes our approach to running ANNS with multi-attribute range filters on the Grant.
We first run a range query of $F$ on the Grant and obtain a set of leaf nodes $U$ that overlap with $F$. 
Then, for each $u \in U$, we run a graph-based ANNS algorithm and pick only objects satisfying $F$.
Note that each leaf node may \textit{partially} overlap $F$, so some objects maintained in a given leaf node may not satisfy $F$.
\Cref{fig:query} depicts an example.

One of our main results is that this approach theoretically sounds and is practically fast, as analyzed below.
Note that $\tilde{O}(\cdot)$ hides any polylog factors.

\begin{algorithm}[!t]
    \caption{\textsc{Query}$(\mathcal{T},\mathbf{q}, F, k)$}  \label{algo:anns}
    \DontPrintSemicolon
    \KwIn {$\mathcal{T}$ (Grant), $\mathbf{q}$, $F$, $k$}
    $S(\mathbf{q},F) \gets \emptyset$\;
    $U \gets$ \textsc{Range-Search}$(\mathcal{T}, F)$\;       \label{algo:anns:range-search}
    \ForEach{$u \in U$}                                     
    {                                                           \label{algo:anns:graph_b}
        $S' \gets$ \textsc{ANNS}$(\mathcal{G},\mathbf{q}, k)$\;
        \ForEach{$x \in S'$ such that $x$ satisfies $F$}{
            Update $S(\mathbf{q},F)$ if necessary             \label{algo:anns:graph_e}
        }
    }
    \textbf{return} $S$
\end{algorithm}

\begin{lemma}   \label{lemma:time-general}
Let $T(n)$ be the time of a range search on $\mathcal{T}$ of $n$ objects.
Our algorithm can run in $T(n) + \tilde{O}(1)$ time.
\end{lemma}

\begin{proof}
Line \ref{algo:anns:range-search} requires $T(n)$ time.
We rely on the recent result \cite{shangqi2026proximity}, and there exists a proximity graph that deals with \textsc{ANNS}$(\mathcal{G},\mathbf{q}, k)$ in $\tilde{O}(1)$ worst time\footnote{In our experiments, we used HNSW as a practical instantiation of \cite{shangqi2026proximity} because it empirically follows similar efficiency trends \cite{malkov2018efficient,prokhorenkova2020graph} while providing superior scalability and implementation simplicity for real-world deployments \cite{wang2021comprehensive,azizi2025graph,li2019approximate}.}.
Since the number leaf nodes is at most $c = O(1)$, $|U| = O(1)$.
Therefore, lines \ref{algo:anns:graph_b}--\ref{algo:anns:graph_e} need $\tilde{O}(1)$ time.
\end{proof}

\vs
\noindent
Existing works that can deal with this problem have no time bound, but our algorithm can yield a sub-linear query time.

\begin{theorem} \label{theorem:time}
If we use a kd-tree as the base tree structure of a Grant, our algorithm can run in $O(n^{1-1/m})$ time.
\end{theorem}

\begin{proof}
A range search on a kd-tree needs $O(n^{1-1/m} + \text{OUT})$ time, where OUT is the result size of the range search \cite{bentley1975multidimensional}.
Then, $T(n) = O(n^{1-1/m} + |U|) = O(n^{1-1/m})$.
Note that $O(n^{1-1/m}) > \tilde{O}(1)$.
From \Cref{lemma:time-general}, this theorem holds.
\end{proof}

\subsection{Optimizing ANNS on Grant by Heuristic}  \label{sec:grant:heuristic}
An interesting observation of the above algorithm is that, although \Cref{theorem:time} suggests that the main cost of our algorithm is the range search theoretically, the main ``practical'' cost is derived from the graph-based ANNS of each leaf node in $U$.
It is important to notice that at most $k$ leaf nodes have the objects $\in S^{*}(\mathbf{q},F)$, but we may have $|U| > k$.
In this case, our algorithm runs graph-based ANNS on at least $|U| - k$ unnecessary leaf nodes.
However, it is hard to know which leaf node has the objects $\in S^{*}(\mathbf{q},F)$ in advance.
We therefore propose heuristic approaches that skip leaf nodes in $U$ that would not have the objects $\in S^{*}(\mathbf{q},F)$.
Note that these heuristics are applied only when $|U| > O(k)$.

\vs
\noindent
\underline{\textbf{Probabilistic node verification.}}
The idea of our heuristic is to make use of an overlapping ratio between $F$ and each leaf node in $U$.
In \Cref{fig:query}, although $u_8$ overlaps $F$, the overlapping space is very slight against the partition of $u_8$ and $F$.
Hence, it is intuitive that $u_8$ would not have the $k$NN objects, suggesting that $u_8$ is skippable.
This idea is specific to the Grant structure.
To implement this idea, we define the overlap ratio between $F$ and a leaf node $u$.

\begin{figure}[!t]
    \centering
    \includegraphics[width=0.40\linewidth]{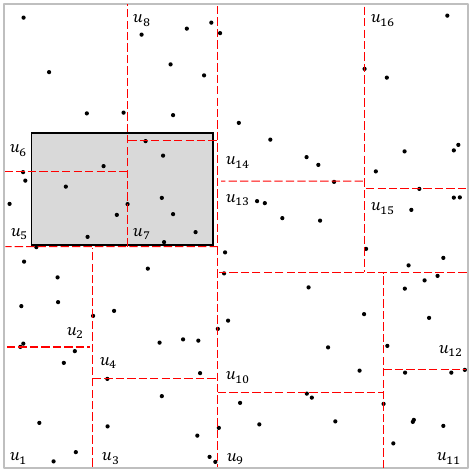}
    \caption{Black points represent two-dimensional attribute values of the objects in $X$, and each partition is maintained by a leaf node $u_i$. 
    The grey rectangle shows $F$, and it overlaps with $u_5$, $u_6$, $u_7$, and $u_8$, so we run graph-based ANNS on these leaf nodes.}
    \label{fig:query}
\end{figure}

\begin{definition}
\textit{Let $V(F)$ be the volume of the hyper-rectangle defined by the ranges in $F$.
Also, let $V(u)$ be the volume of the hyper-rectangle of a node $u$.
The overlap ratio between $F$ and $u$, $r(F, u)$, is}
\begin{equation*}
    r(F, u) = \max \left( \frac{V(F \cap u)}{V(F)}, \frac{V(F \cap u)}{V(u)}\right),
\end{equation*}
\textit{where $F \cap u$ is the overlapping space between $F$ and $u$.}
\end{definition}

Assuming a local uniform distribution of attributes within a leaf node, the number of valid candidates $E[n_{valid}]$ is estimated as
$$E[n_{valid} \,|\, u, F] \propto \frac{n}{c} \cdot r(F, u).$$
Nodes with higher $E[n_{valid}]$ can be prioritized to increase the likelihood of retrieving true nearest neighbors.

For $u \in U$, we have $r(F, u) \in (0, 1]$, and $r(F, u)$ is obtained in $O(1)$ time.
Because it is not guaranteed that $k$ leaf nodes with the largest overlap ratio have the $k$NN objects, we employ a probabilistic approach.
Specifically, for each leaf node $u$ in the range search result, if $r(F, u)$ is larger than a random real value $\in (0,1]$, we run graph-based ANNS in this leaf node.

\subsection{Workload-aware ANNS}    \label{sec:grant:workload}
If past query vectors are available, we can enhance our proposed approach.
Let $\mathbf{p}$ be the nearest vector to a given query vector $\mathbf{q}$ among a set of past query vectors\footnote{This is a different assumption from that of \cite{mohoney2023high,li2025sieve}, since they assume that not past query vectors but past query \textit{filters} are available.
Our assumption is easier to hold: for example, we can use the vectors in $X$ as past query vectors.}.
Furthermore, let $dist(\mathbf{p}, u_i)$ be the distance between $\mathbf{p}$ and the nearest neighbor vector in the leaf node $u_i$ of the Grant.
For ease of explanation, assume that $dist(\mathbf{p}, u_1) < dist(\mathbf{p}, u_2) < \cdots dist(\mathbf{p}, u_{c})$.
Assuming that $k = 2$ and $F$ overlaps with 10 leaf nodes, $u_1$, $u_{2}$, ..., and $u_{10}$, it is intuitively known that running ANNS of $\mathbf{q}$ in $u_1$ and $u_2$ is sufficient to obtain the objects $\in S^{*}(\mathbf{q},F)$.
This is because, if $dist(\mathbf{p}, \mathbf{q})$ is small, $\mathbf{p}$ and $\mathbf{q}$ share similar ANNS results.
This idea also skips ANNS on unnecessary nodes and enables accessing only $O(k)$ leaf nodes while keeping high recall.  

\vs
\noindent
\textbf{Additional pre-processing.}
Let $\mathbf{P}$ be a set of past query vectors, and it is maintained based on a graph index $\mathcal{G}_{\mathbf{P}}$.
Additionally, for each past query vector $\mathbf{p} \in \mathbf{P}$, we run ANNS on each leaf node $u$ of the Grant.
Let $x$ be the object obtained for the ANNS, and we record $dist(\mathbf{p},\mathbf{v})$.
We use a matrix $\mathcal{M}$ to maintain such a record:
\begin{equation*}
    \mathcal{M} =
\begin{pmatrix}
    dist(\mathbf{p}_1, \mathbf{v}_{\alpha})      & \cdots    & dist(\mathbf{p}_1, \mathbf{v}_{\beta})         \\
    \vdots                              & \ddots    & \vdots                                \\
    dist(\mathbf{p}_{c}, \mathbf{v}_{\gamma})    & \cdots    & dist(\mathbf{p}_{c}, \mathbf{v}_{\phi})
\end{pmatrix}
\end{equation*}
where the $i$-th row $j$-th column is the ANN result of $\mathbf{p}_{i}$ in the leaf node $u_j$.
This is done after running building Grant.
Note that this additional pre-processing is also done only once.

When an HNSW is used as the index of $\mathbf{P}$ such that $|\mathbf{P}| \leq n$, the additional pre-processing time is $\tilde{O}(|\mathbf{P}|n)$.
This is because building an HNSW of $\mathbf{P}$ needs $O(|\mathbf{P}|\log|\mathbf{P}|)$ time \cite{malkov2018efficient} and obtaining $\mathcal{M}$ requires $O(|\mathbf{P}|c) = O(|\mathbf{P}|)$ time.
Similarly, its space complexity is $O(|\mathbf{P}|n)$.
(To reduce the time and space complexities, we can randomly sample past queries.)

\begin{algorithm}[!t]
    \caption{\textsc{Query}$(\mathcal{T},\mathcal{G}_{\mathbf{P}},\mathcal{M},\mathbf{q}, F, k)$}  \label{algo:anns-w}
    \DontPrintSemicolon
    \KwIn {$\mathcal{T}$ (Grant), $\mathcal{G}_{\mathbf{P}}$, $\mathcal{M}$, $\mathbf{q}$, $F$, $k$}
    $S(\mathbf{q},F) \gets \emptyset$\;
    $U, U' \gets$ \textsc{Range-Search}$(\mathcal{T}, F)$\;
    \If {$|U'| > O(k)$}
    {
        $\mathbf{p}_i \gets$ \textsc{ANNS}$(\mathcal{G}_{\mathbf{P}},\mathbf{q}, 1)$\;
        $U \gets  O(k)$ nodes in $U'$ such that $dist(\mathbf{p}_i,\cdot)$ (in $\mathcal{M}[i]$) is the smallest
    }
    Run lines \ref{algo:anns:graph_b}--\ref{algo:anns:graph_e} of \Cref{algo:anns}\;
    return $S$
\end{algorithm}

\vs
\noindent
\textbf{Query processing.}
\Cref{algo:anns-w} presents our algorithm for ANNS with multi-attribute range filters on Grant with $\mathcal{G}_{P}$ and $\mathcal{M}$.
Its main difference from \Cref{algo:anns} is to exploit a past query vector that is similar to a given query vector $\mathbf{q}$.
We obtain a set $U'$ of leaf nodes that overlap $F$.
If $|U'|$ is large, we compute $\mathbf{p}_i$, which is an ANN vector to $\mathbf{q}$ in $\mathbf{P}$, by running graph-based ANN on $\mathcal{G}_{\mathbf{P}}$.
Then, we pick only $O(k)$ leaf nodes in $U'$ such that $dist(\mathbf{p}_i,\cdot)$ is the smallest.
The subsequent operations are the same as our original query processing algorithm.
This is an optional workload-aware optimization that can be applied on top of Grant when past query vectors are available.

\section{Experiment}    \label{sec:experiment}
All experiments were conducted on a Ubuntu 22.04 LTS machine with a 2.2GHz Intel Core i9-13950HX processor and 128GB RAM.

\subsection{Setting}
\underline{\textbf{Datasets.}}
We used four datasets commonly used in ANNS works: Deep \cite{babenko2016efficient}, SIFT (\url{http://corpus-texmex.irisa.fr/}), SPACEV \cite{chen2021spann}, and T2 (\url{https://big-ann-benchmarks.com/neurips21.html}). 
They are large-scale and contain 10 million vectors.
Deep is a set of 96-dimensional deep features, and SIFT is a set of 128-dimensional SIFT features.
SPACEV is a set of 100-dimensional learned document vectors, and T2I is a set of 200-dimensional vectors learned from the Yandex visual search engine results.
We followed the existing filtered ANNS works \cite{mohoney2023high,wei2020analyticdb,wang2021milvus,li2025sieve} to provide each object with attribute values, and the attribute values follow a power-law distribution \cite{gupta2023caps}, which reflects real-world scenarios \cite{petersen2016power}.
By default, $m = 2$.

\vs
\noindent
\underline{\textbf{Queries.}}
For each dataset, we used 100 random query vectors that can be obtained from the same source of the dataset.
Given a query vector $\mathbf{q}$, to generate query range filters $F$, we first assigned a random value, which also follows the same power-law distribution of the attribute domain, to $l_j$ for each $j \in [1,m]$.
Then, $u_j = l_j +$ $\lambda \times$ domain size, where $\lambda \in [0.2, 0.8]$.
By default, $\lambda = 0.2$ and $k = 10$.
In this default setting, the average selectivity (i.e., the ratio of the objects satisfying the range filter) was 0.15.

\vs
\noindent
\underline{\textbf{Evaluated algorithms.}}
For algorithms that do not require past query workloads, we evaluated the following ones:
\begin{itemize}
    \setlength{\leftskip}{-4.0mm}
    \item   \textsf{BF}: A brute-force exact algorithm.
    \item   \textsf{Pre-filter}: This algorithm employs a kd-tree based on the attribute values of the objects in $X$.
            Given a query, it first traverses the kd-tree based on the query range filters and obtains a set of objects that satisfy the range filters.
            Then, it runs a linear scan of this set to return the $k$NN vectors, so this algorithm also returns the exact answer.
            This algorithm actually corresponds Milvus \cite{wang2021milvus}.
            Navix \cite{sehgal2025navix} also follows this approach but returns an approximate result.
            We tested Navix, but its performance was worse than the exact version of Pre-filter, so we omit its result.
    \item   \textsf{IVF}: An ANNS algorithm that is based on an inverted file structure.
            This approach is widely employed in the state-of-the-art ANNS systems \cite{wang2021milvus,wei2020analyticdb,mohoney2023high}.
    \item   \textsf{Post-filter}: The baseline algorithm introduced in \Cref{sec:preliminary}.
            We used HNSW \cite{malkov2018efficient} as its ANNS index.
            This algorithm actually corresponds to SIEVE \cite{li2025sieve} without past query filters (see \Cref{sec:related-work}) and \textsc{VBase} \cite{zhang2023vbase}.
    \item   \textsf{ACORN}: The baseline algorithm introduced in \Cref{sec:preliminary}.
            Because ACORN-$\gamma$ \cite{patel2024acorn} could not be built in a reasonable time (i.e., 1 hour), we used ACORN-1.
    \item   \textsf{Grant-basic}: Our algorithm described in \Cref{sec:grant:query}.
            We used a kd-tree as the base tree structure, and each leaf node maintains the corresponding objects by HNSW.
            This setting was applied to the next algorithms.
            We set $c = 150$ (see our supplementary).
    \item   \textsf{Grant-pnv}: Our algorithm with the probabilistic node verification heuristic.
\end{itemize}
For algorithms that require past query workloads, we evaluated
\begin{itemize}
    \setlength{\leftskip}{-4.0mm}
    \item   \textsf{WA-Grant}: Our algorithm that exploits past query vectors, see \Cref{sec:grant:workload}.
            For each dataset, we randomly sampled 1,000 vectors and added Gaussian noise to them.
            We used them as past query vectors, and they were also indexed by an HNSW.
    \item   \textsf{SIEVE} \cite{li2025sieve}: the state-of-the-art algorithm for ANNS with range filter on multiple attributes.
            To maximize performance, we used the same query filters as past ones (although it is impractical).
            Because this algorithm requires a memory budget, we set it so that its memory usage is almost the same as \textsf{WA-Grant} for fair comparison.
\end{itemize}
All algorithms were implemented in C++ and compiled by g++ 11.3.0 with -O3 optimization.
They enabled SIMD instructions.

For the algorithms using HNSW, we used the original implementation\footnote{\url{https://github.com/nmslib/hnswlib}}.
We tuned the inner parameters to obtain a Pareto-optimal QPS-recall tradeoff.
For HNSW, we did a grid search from $M=\{8, 16, 32, 64, 128\}$ and $ef_{construction} =\{32, 64, 100, 200, 300\}$.
As a result, we set $M = 64$ and $ef_{construction} = 200$ for the baselines, whereas we set $M = 16$ and $ef_{construction} = 64$ for Grant.
(Since each leaf node has at most $n/c < n$ vectors, each HNSW in Grant has smaller parameter values.)
For IVF, we used 1,000 clusters.

We did not consider works that assume matching filters on multiple attributes \cite{gupta2023caps,wang2023efficient}, as they are not trivial to efficiently address our problem.
A straightforward extension of matching filter techniques to range filters requires discretizing each continuous range into a very large number of (in principle, infinitely many) equality predicates, which is impractical.

\subsection{Result}
We first report the performances of the algorithms that do not require past query workloads.
For indexing time, memory usage, impact of $c$, impact of $k$, impact of $m$, and impact of range size, refer to our supplementary. 

\vs
\noindent
\underline{\textbf{Impact of $c$.}}
We tuned $c$ for Grant, and \Cref{fig:c} shows QPS (queries per second) under recall $\approx 0.9$ on SIFT and SPACEV.
Note that (i) we have similar results on Deep and T2I and (ii) the ``$< 76$'' case shows the best result when $c \leq 75$.

Although Grant is not very sensitive to $c$, it has an optimal $c$.
When $c$ is large, we have many leaf nodes overlapping $F$, so it incurs many ANNSs.
On the other hand, when $c$ is small, we have fewer leaf nodes overlapping $F$ but have a larger number of objects that do not satisfy $F$.
These observations yield the results in \Cref{fig:c}.

\vs
\noindent
\underline{\textbf{Indexing time.}}
\Cref{tab:offline} shows the indexing time of each algorithm.
(Since \textsf{BF} does not employ any data structure, it is omitted.)
This phase used 24 threads.
\textsf{Grant-basic} and \textsf{Grant-pnv} share the same data structure, so their indexing time is represented by \textsf{Grant}.
Since \textsf{Pre-filter} employs a single tree structure based only on attributes, it needs a shorter indexing time than the others.

Although \textsf{Grant} employs tree and graph structures, its indexing time is at most just 1.5 minutes for 10 million high-dimensional vectors.
An interesting observation is that \textsf{Grant} needs shorter indexing time than \textsf{IVF}, \textsf{Post-filter}, and \textsf{ACORN}.

\begin{figure}[!t]
    \begin{center}
        \subfigure[SIFT]{%
            \includegraphics[width=0.35\linewidth]{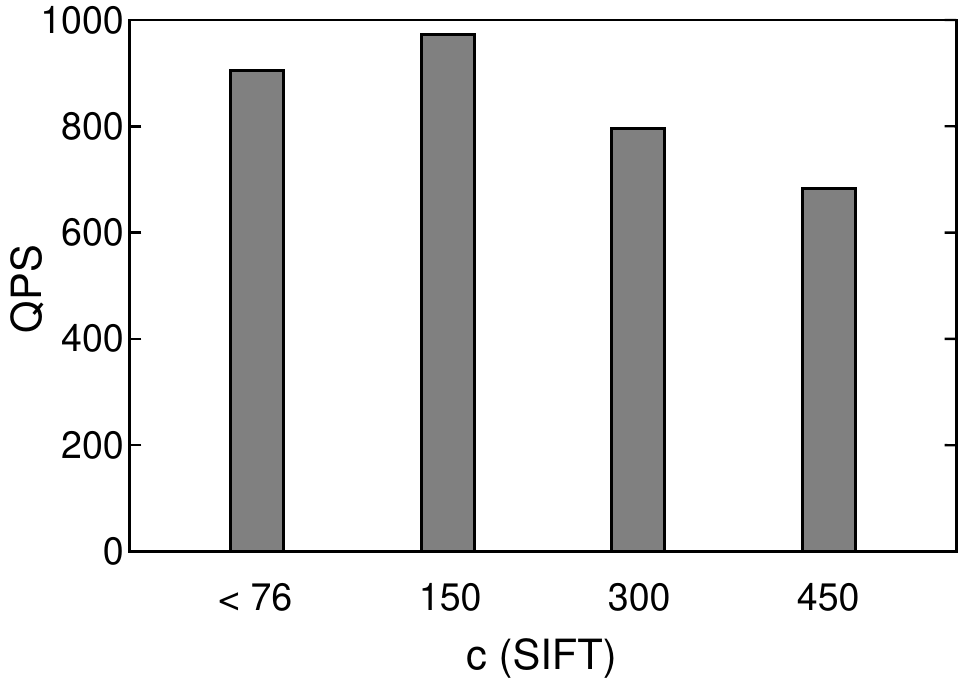}    \label{fig:sift_c}}
        \subfigure[SPACEV]{%
            \includegraphics[width=0.35\linewidth]{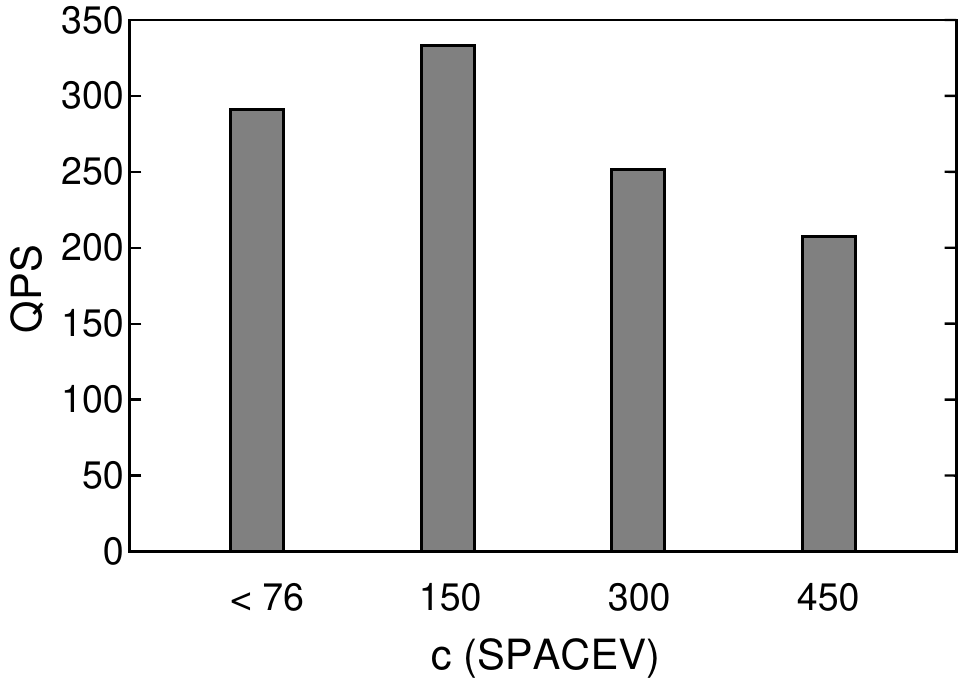}  \label{fig:spacev_c}}
        \caption{Impact of $c$ under recall $\approx 0.9$}
        \label{fig:c}
    \end{center}
\end{figure}

\begin{table}[!t]
    \centering
    \caption{Indexing time [sec]}
    \label{tab:offline}
    \begin{tabular}{lrrrr} \toprule
        Dataset                 & Deep      & SIFT      & SPACEV    & T2I       \\ \midrule
        \textsf{Pre-filter}     & 3.04      & 3.04      & 3.03      & 3.06      \\ 
        \textsf{IVF}            & 98.75     & 151.92    & 101.37    & 202.75    \\
        \textsf{Post-filter}    & 888.26    & 1109.36   & 1380.33   & 1040.26   \\
        \textsf{ACORN}          & 1666.40   & 2067.56   & 2044.77   & 2013.43   \\
        \textsf{Grant}          & 60.20     & 71.52     & 78.34     & 97.14     \\
        \bottomrule
    \end{tabular}
\end{table}
\begin{table}[!t]
    \centering
    \caption{Memory consumption [GB]}
    \label{tab:memory}
    \begin{tabular}{lrrrr} \toprule
        Dataset                 & Deep  & SIFT  & SPACEV    & T2I   \\ \midrule
        \textsf{Pre-filter}     & 0.89  & 0.89  & 0.89      & 0.89  \\ 
        \textsf{IVF}            & 0.13  & 0.12  & 0.13      & 0.12  \\
        \textsf{Post-filter}    & 9.82  & 11.01 & 9.97      & 13.70 \\
        \textsf{ACORN}          & 9.82  & 11.01 & 9.96      & 13.69 \\
        \textsf{Grant}          & 6.66  & 7.85  & 6.81      & 10.53 \\
        \bottomrule
    \end{tabular}
\end{table}

\vs
\noindent
\underline{\textbf{Memory consumption.}}
\Cref{tab:memory} exhibits the memory consumption of the data structure of each algorithm.
Because of the graph-based indices, \textsf{Post-filter} and our algorithm needs more memory than \textsf{Pre-filter} and \textsf{IVF}.
Nonetheless, our algorithm does not consume much memory, and only less than about 10GB is required for 10 million high-dimensional vectors, which is easy to fit into modern servers.
(\textsf{ACORN} employs a graph index that is essentially HNSW, so its memory consumption is almost the same as \textsf{Post-filter}.)

\vs
\noindent
\underline{\textbf{QPS vs. recall.}}
We compare our algorithms with the baseline ones w.r.t. time and accuracy tradeoff.
\Cref{fig:tradeoff} plots the QPS vs. recall curve of each algorithm.
Note that QPS and recall were the averages of the results of the 100 queries.
We used a single thread for each query.

Our algorithms generally have better time-recall tradeoffs than the existing state-of-the-art ones on the four datasets.
The exact algorithms (i.e., \textsf{BF} and \textsf{Pre-filter}) are slow.
\textsf{IVF}, \textsf{Post-filter}, and \textsf{ACORN} improve the efficiency compared with the exact algorithms.
However, their search spaces are still large and suffer from many unnecessary object accesses.
Thanks to the Grant structure, our algorithms alleviate this issue, yielding superiority against the existing techniques.
They obtain high-recall results (e.g., 0.9) within a few milliseconds, which is also seen in \Cref{tab:decomposed} and is usually required in real-world systems \cite{zhang2023vbase}.

\begin{figure*}[!t]
    \begin{center}
        \subfigure[Deep]{%
            \includegraphics[width=0.23\linewidth]{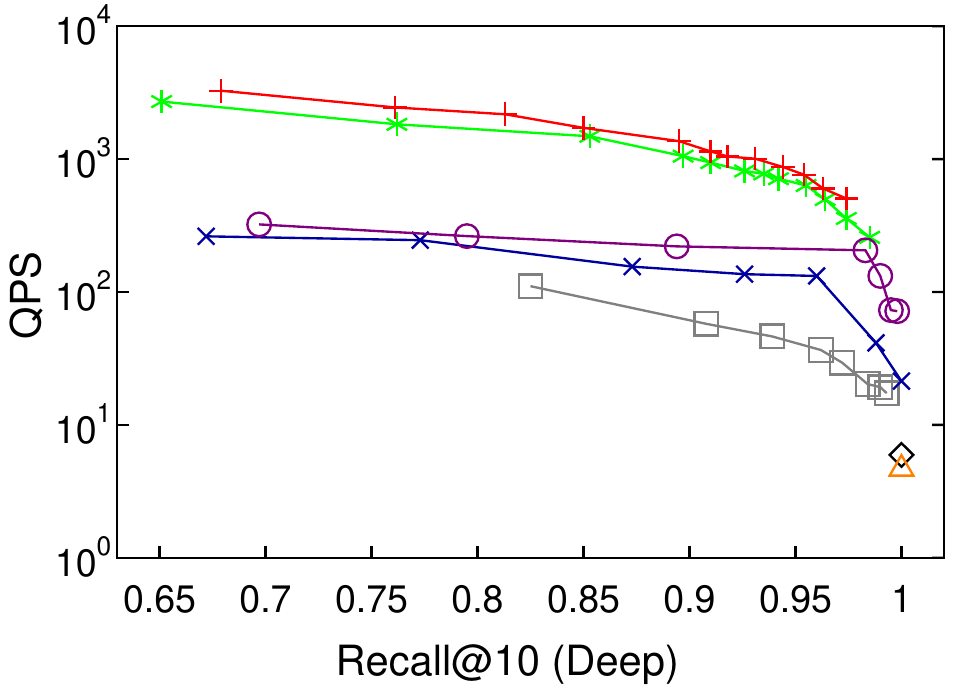}    \label{fig:deep}}
        \subfigure[SIFT]{%
            \includegraphics[width=0.23\linewidth]{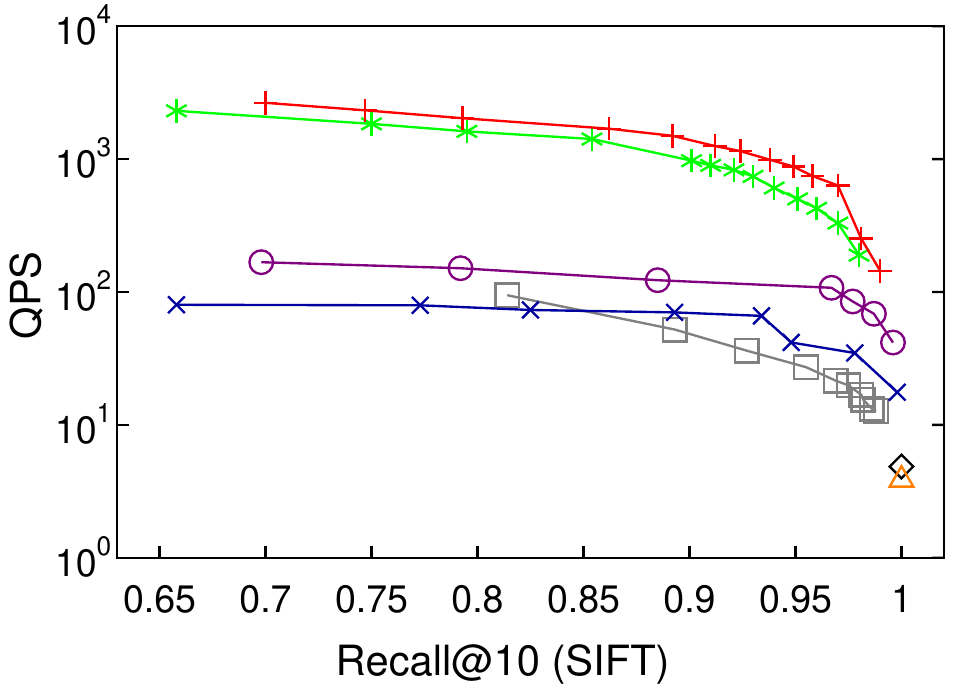}    \label{fig:sift}}
        \subfigure[SPACEV]{%
            \includegraphics[width=0.23\linewidth]{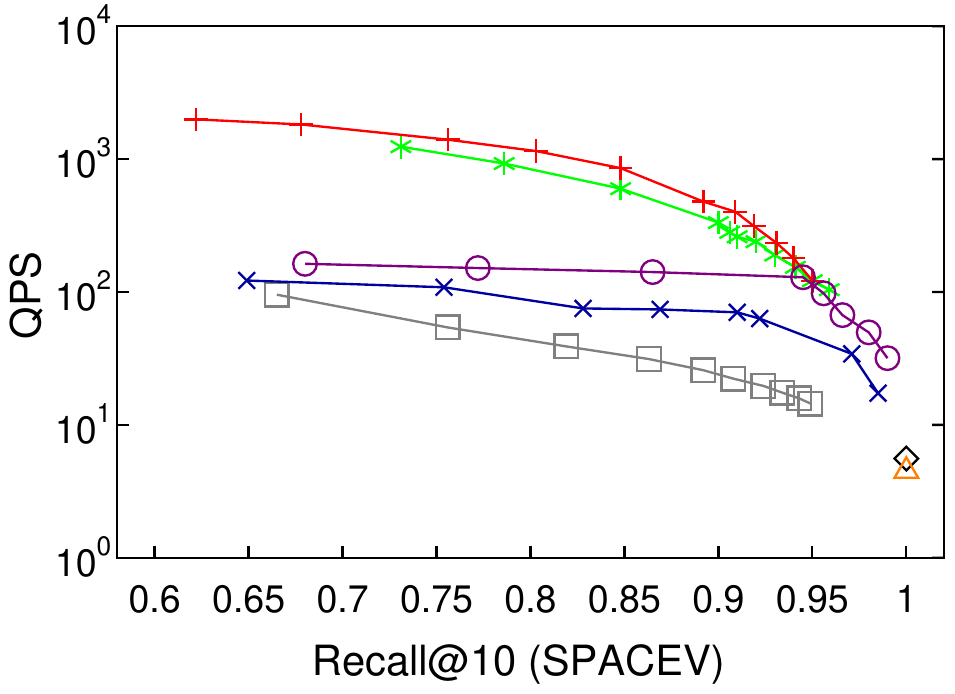}  \label{fig:spacev}}
        \subfigure[T2I]{%
            \includegraphics[width=0.23\linewidth]{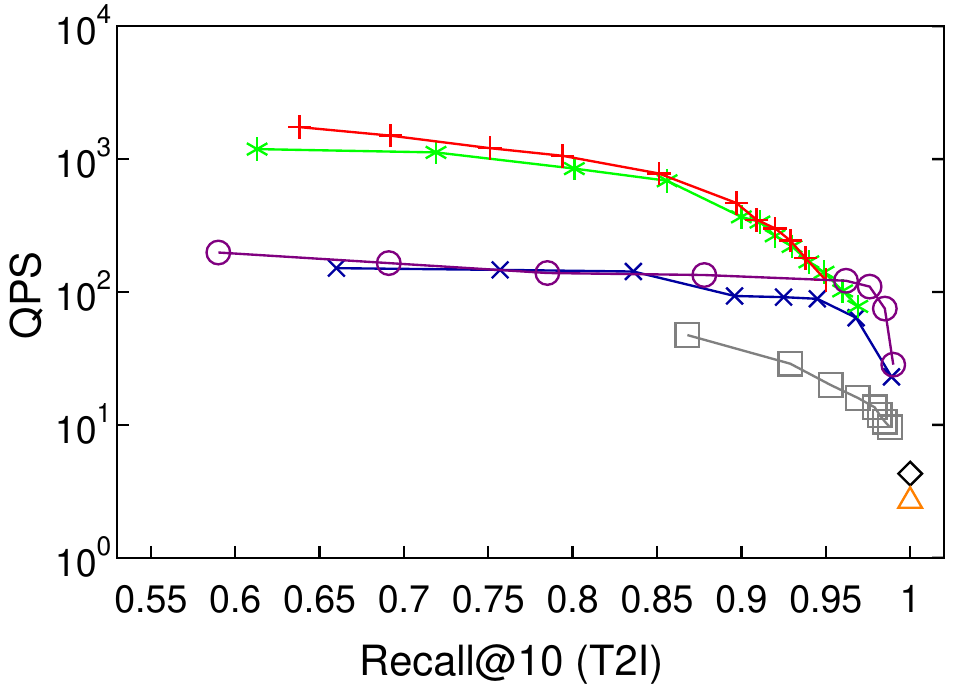}     \label{fig:t2i}}
        \caption{QPS vs. recall@10.
            ``$\diamond$'' shows \textsf{BF}, ``\textcolor{orange}{$\triangle$}'' shows \textcolor{orange}{\textsf{Pre-filter}}, ``\textcolor{gray}{$\square$}'' shows \textcolor{gray}{\textsf{IVF}}, ``\textcolor{blue}{$\times$}'' shows \textcolor{blue}{\textsf{Post-filter}}, ``\textcolor{violet}{$\circ$}'' shows \textcolor{violet}{\textsf{ACORN}}, ``\textcolor{green}{$\ast$}'' shows \textcolor{green}{\textsf{Grant-basic}}, and ``\textcolor{red}{$+$}'' shows \textcolor{red}{\textsf{Grant-pnv}}.}
        \label{fig:tradeoff}
    \end{center}
\end{figure*}

\vs
\noindent
\underline{\textbf{Decomposed time}.}
To justify the motivation in \Cref{sec:grant:heuristic}, we show the decomposed time of our algorithms in \Cref{tab:decomposed}.
Note that \textsf{RS} represents \textsc{Range-Search} in \Cref{algo:anns}.
\textsf{Grant-basic} and \textsf{Grant-pnv} share the same \textsf{RS}.
\Cref{tab:decomposed} confirms that the practical bottleneck of our algorithms is not \textsf{RS} but graph-based ANNS.
\Cref{fig:tradeoff} and \Cref{tab:decomposed} describe that \textsf{Grant-pnv} relieves this bottleneck and is better than \textsf{Grant-basic}.
\textsf{Grant-pnv} yields about 1.3--1.5x efficiency improvement against \textsf{Grant-basic}.
This result successfully demonstrates that the heuristic provides efficiency acceleration without recall drop.

\begin{figure*}[!t]
    \begin{center}
        \subfigure[Deep]{%
            \includegraphics[width=0.23\linewidth]{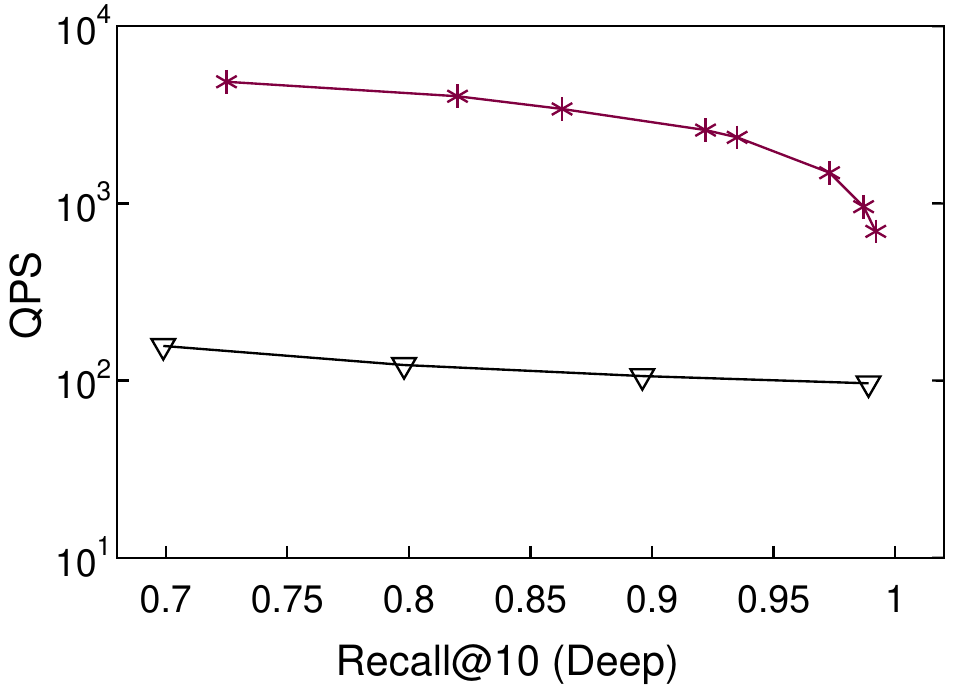}    \label{fig:deep_q}}
        \subfigure[SIFT]{%
            \includegraphics[width=0.23\linewidth]{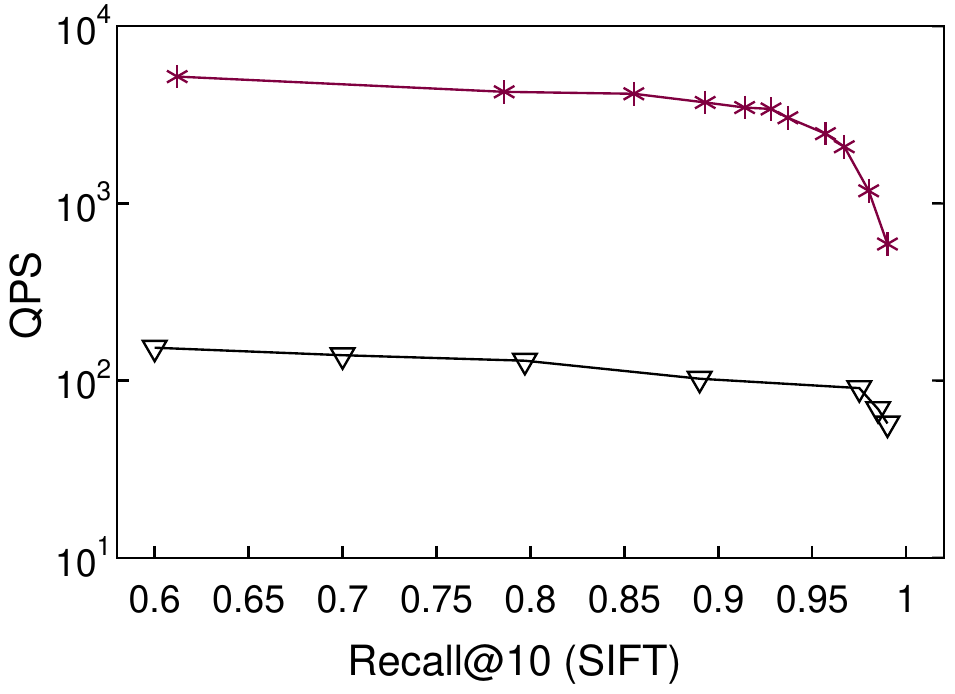}    \label{fig:sift_q}}
        \subfigure[SPACEV]{%
            \includegraphics[width=0.23\linewidth]{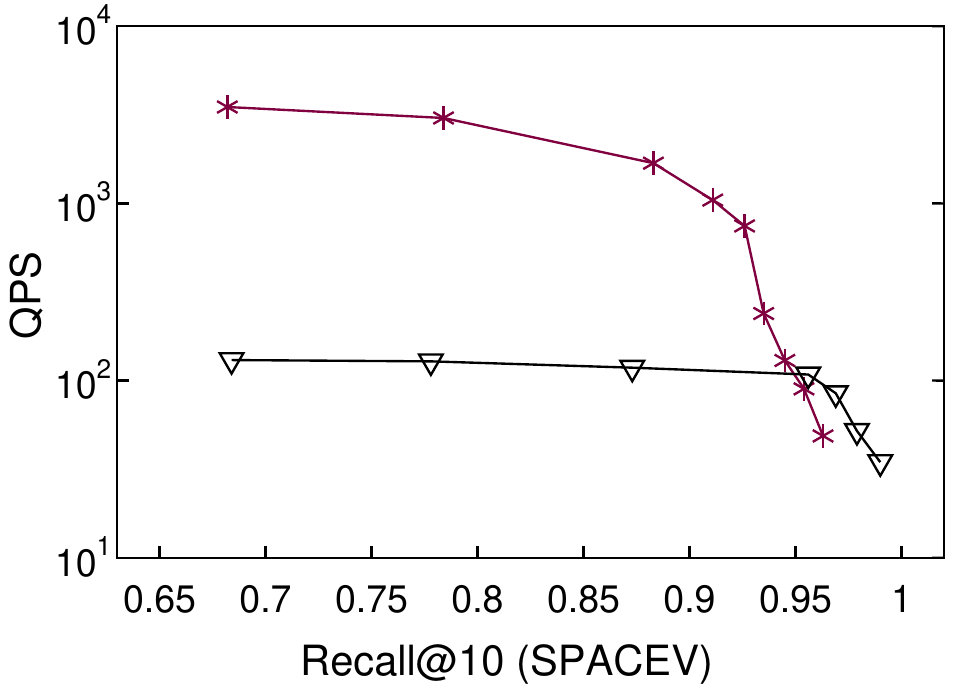}  \label{fig:spacev_q}}
        \subfigure[T2I]{%
            \includegraphics[width=0.23\linewidth]{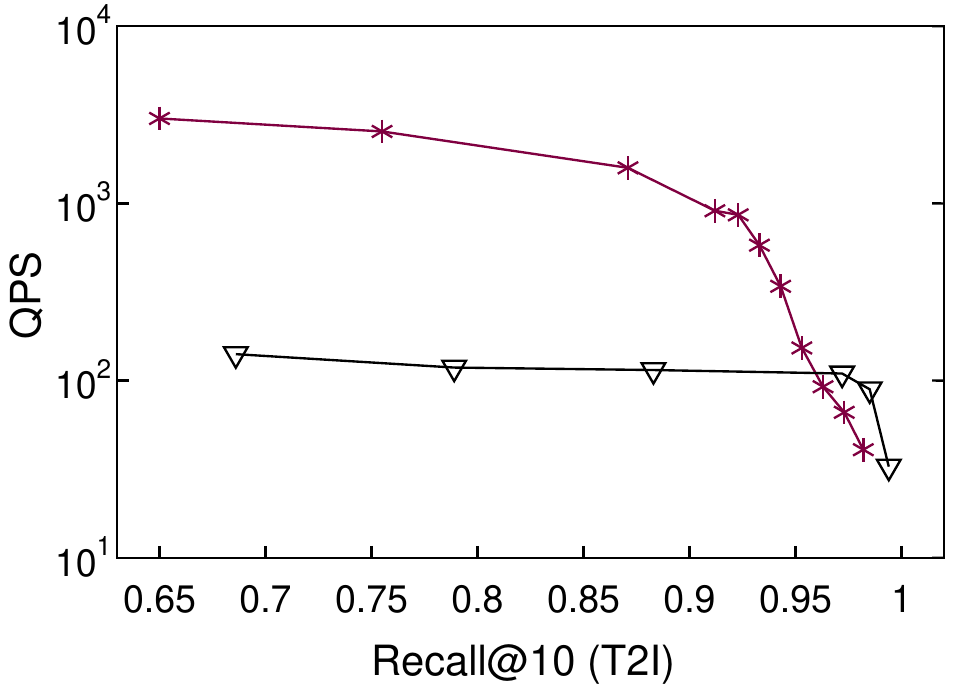}     \label{fig:t2i_q}}
        \caption{QPS vs. recall@10.
            ``$\triangledown$'' shows \textsf{SIEVE} and ``\textcolor{purple}{$*$}'' shows \textcolor{purple}{\textsf{WA-Grant}}.}
        \label{fig:tradeoff_q}
    \end{center}
\end{figure*}
\begin{table}[!t]
    \centering
    \caption{Decomposed time [microsec] (recall $\approx 0.9$). \textsf{RS} represents \textsc{Range-Search} in \Cref{algo:anns}, whereas \textsf{Grant-basic} and \textsf{Grant-pnv} show the running time in lines \ref{algo:anns:graph_b}--\ref{algo:anns:graph_e} of \Cref{algo:anns}.}
    \label{tab:decomposed}
    \begin{tabular}{lrrrr} \toprule
        Dataset                 & Deep      & SIFT      & SPACEV    & T2I       \\ \midrule
        \textsf{RS}             & 0.60      & 0.50      & 0.87      & 0.88      \\ 
        \textsf{Grant-basic}    & 949.06    & 1041.10   & 3839.24   & 2718.39   \\
        \textsf{Grant-pnv}      & 700.15    & 667.52    & 2486.83   & 2137.70   \\
        \bottomrule
    \end{tabular}
\end{table}

\vs
\noindent
\underline{\textbf{QPS vs. recall of WA-Grant and SIEVE.}}
\Cref{fig:tradeoff_q} compares \textsf{WA-Grant} and \textsf{SIEVE} under the same memory budget, and we see that \textsf{WA-Grant} normally shows a better speed-accuracy tradeoff.
For example, at recall@10 = 0.9, a reasonable accuracy, \textsf{WA-Grant} is substantially better than \textsf{SIEVE}.
This result demonstrates that our idea of limiting nodes for graph traversal based on past query \textit{vectors} is effective.
Note that the performance of \textsf{WA-Grant} is better than that of \textsf{Grant-pnv} (more than 2 times faster at recall@10 = 0.9), confirming that Grant can be accelerated by using past query vectors.
(This observation can be seen by comparing the results in \Cref{fig:tradeoff,fig:tradeoff_q}.)

\section{Conclusion}    \label{sec:conclusion}
This paper addressed the problem of ANNS with multi-attribute range filters and proposed Grant, a new data structure that solves this problem efficiently.
Our theoretical analysis proves the efficiency of Grant.
In addition, we proposed two heuristics that optimized the query processing performance of Grant.
Our experiments demonstrate that Grant outperforms the existing state-of-the-art techniques.

\vs
\noindent
\textbf{Limitation.}
This work considers a static dataset, and how to efficiently deal with dynamic data (i.e., vector insertions and deletions) is an open problem.
In addition, this work assumes that $m$ is small because real-world datasets have this tendency (e.g., price, rating, and date), and \textit{it is highly unlikely that ranges are specified for many attributes}.
If this assumption does not hold, the Grant structure may not be efficient because the tree structure is efficient only when $m$ (or dimensionality) is small (although, due to the curse of dimensionality, it is common to degrade the performance of a given algorithm for large $m$).

\vs
\noindent
\textbf{Acknowledgment.}
We thank Dr. Yusuke Matsui for his useful comments, which improved the presentation of this paper.

\appendix
\section{Additional Experiments}
\underline{\textbf{Impact of $k$.}}
\Cref{fig:k} illustrates the result under recall $\approx 0.9$.
We omit \textsf{Grant-basic} because \textsf{Grant-pnv} is always better.
The QPS of \textsf{Grant-pnv} decreases as $k$ increases.
Nevertheless, it is much better than the existing algorithms.
\textsf{BF}, \textsf{Pre-filter}, and \textsf{IVF} are not affected by $k$ because they access all objects in $X$, all objects in a result set of the range search, and all objects in the corresponding posting lists, respectively.
As $k$ becomes larger, the QPS of \textsf{Post-filter} becomes very low (e.g., less than 40 for $k = 50$) because \textsf{Post-filter} needs to iterate more ANNS.
\textsf{ACORN} failed to obtain nearly 0.9 recall (so it does not have plots) when $k \geq 20$, because its graph traversal approach loses the path to reach the exact $k$NN objects.

\vs
\noindent
\underline{\textbf{Impact of $m$.}}
We next study the impact of the number of attributes $m$.
\Cref{fig:m} shows that our algorithm keeps outperforming the other ones.
Recall the time complexity of our algorithm, and its empirical performance follows the theoretical result.
We see that \textsf{Post-filter} and \textsf{Pre-filter} have the opposite tendency.
\textsf{Post-filter} is very sensitive to $m$, whereas \textsf{Pre-filter} improves its efficiency with a larger $m$.
This is because a larger $m$ usually yields a smaller range search result set.
\textsf{ACORN} again failed to obtain nearly 0.9 recall when $m \geq 3$.
These results demonstrate that \textsf{ACORN} is not appropriate for our problem.

\vs
\noindent
\underline{\textbf{Impact of range size.}}
To investigate the impact of range size (i.e., $|X(F)|$), we ran experiments by varying $\lambda$.
\Cref{fig:range} exhibits the robustness of our algorithm against range size.
This robustness is attributed to their unnecessary leaf node skipping heuristics.
A larger range size provides a larger $X(F)$, so the performance of \textsf{Pre-filter} (\textsf{Post-filter}) becomes lower (higher) as $\lambda$ becomes larger.
The performance of \textsf{Post-filter} did not reach that of our algorithm even when $\lambda$ is large.

\begin{figure*}[!t]
    \begin{center}
        \subfigure[Deep]{%
            \includegraphics[width=0.23\linewidth]{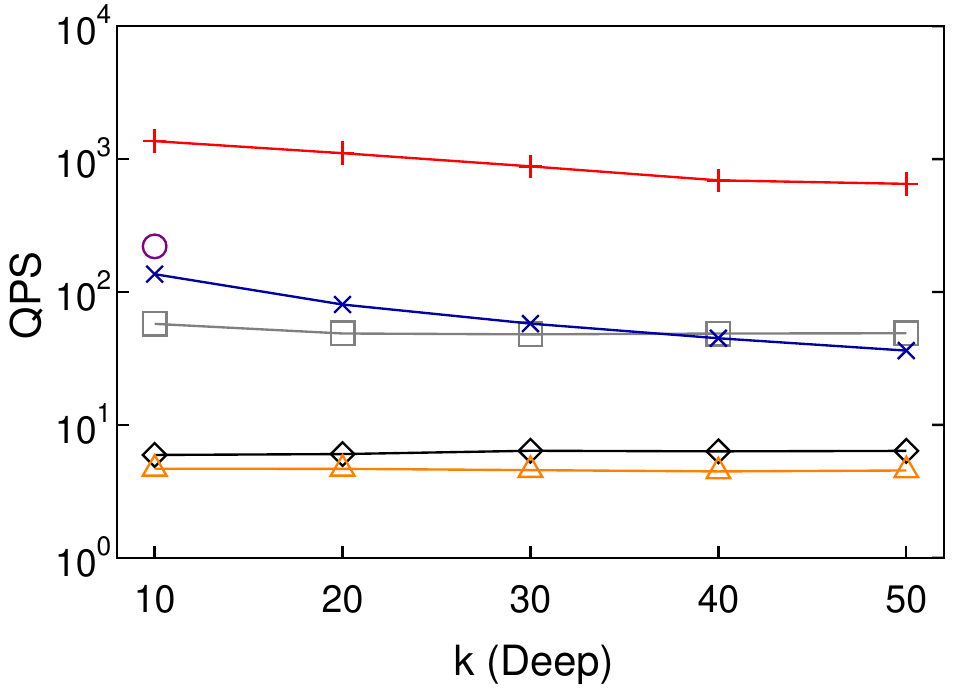}    \label{fig:deep_k}}
        \subfigure[SIFT]{%
            \includegraphics[width=0.23\linewidth]{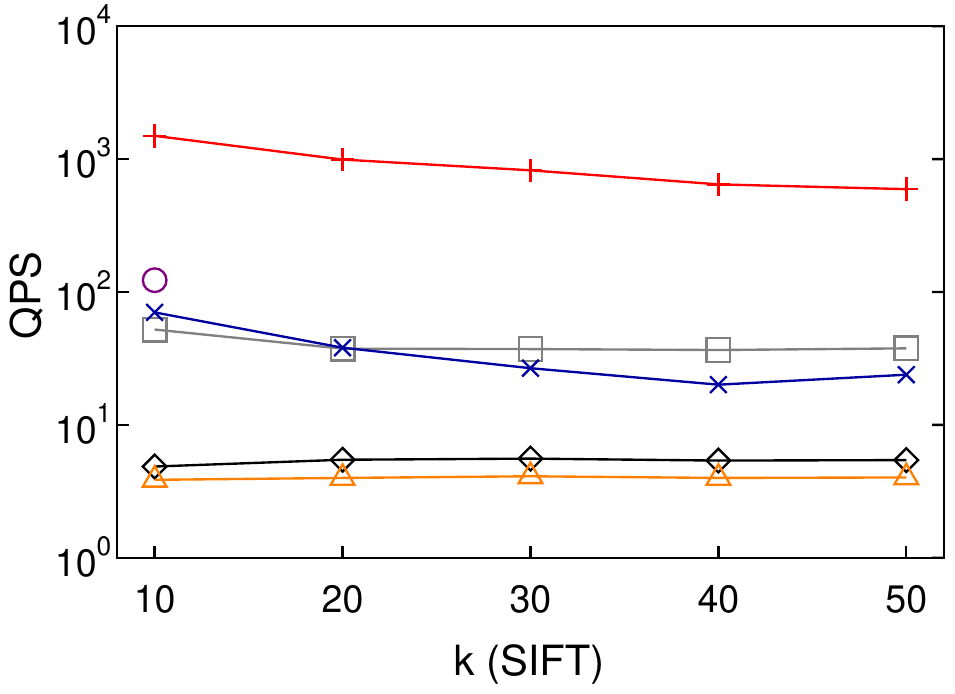}    \label{fig:sift_k}}
        \subfigure[SPACEV]{%
            \includegraphics[width=0.23\linewidth]{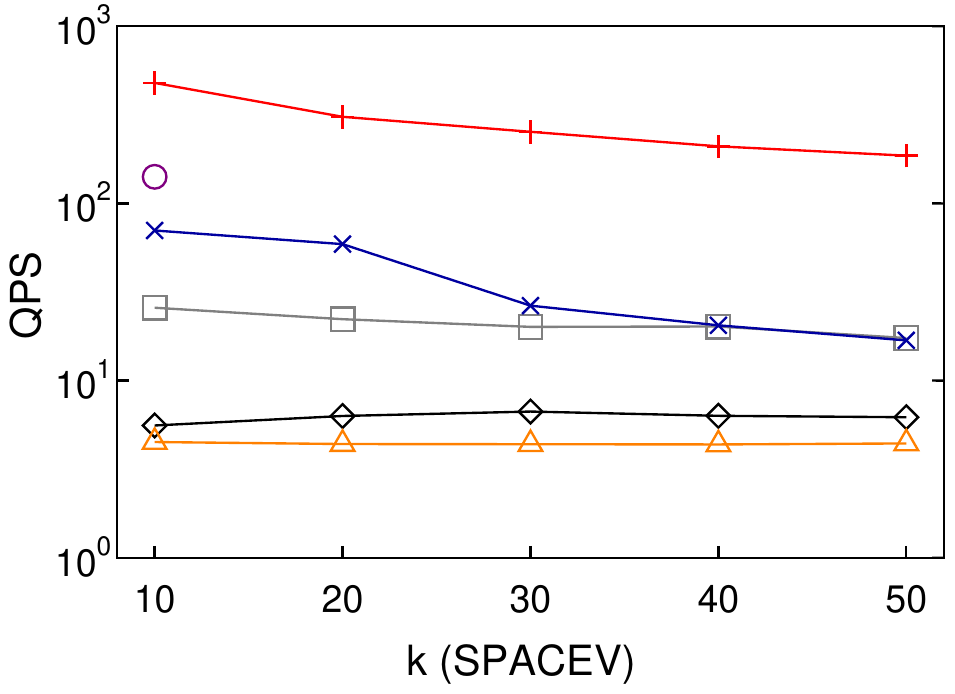}  \label{fig:spacev_k}}
        \subfigure[T2I]{%
            \includegraphics[width=0.23\linewidth]{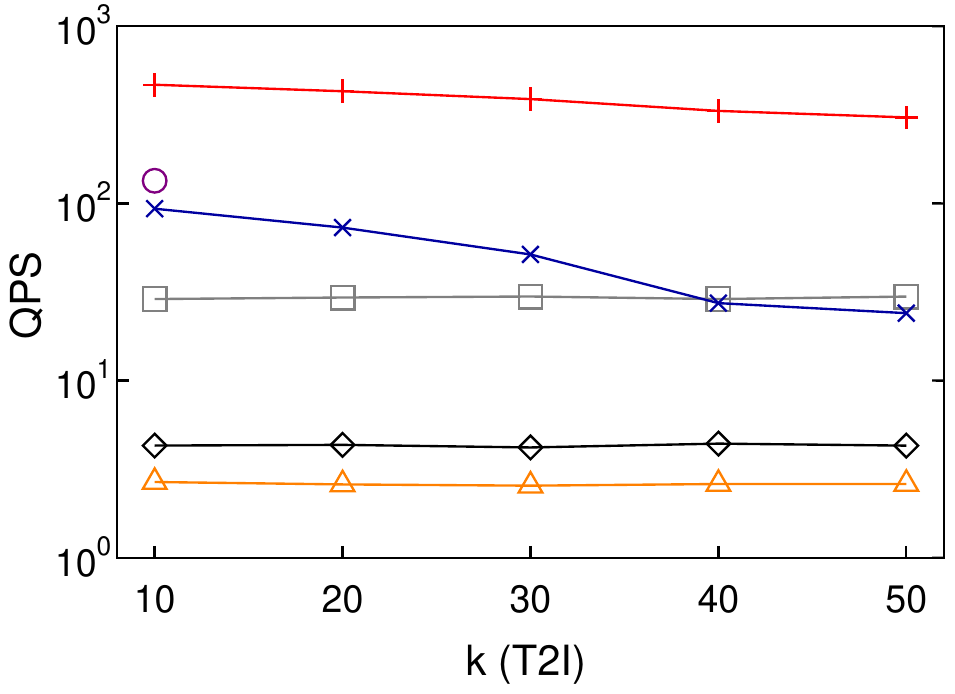}     \label{fig:t2i_k}}
        \caption{Impact of $k$ under recall $\approx 0.9$.
            ``$\diamond$'' shows \textsf{BF}, ``\textcolor{orange}{$\triangle$}'' shows \textcolor{orange}{\textsf{Pre-filter}}, ``\textcolor{gray}{$\square$}'' shows \textcolor{gray}{\textsf{IVF}}, ``\textcolor{blue}{$\times$}'' shows \textcolor{blue}{\textsf{Post-filter}}, ``\textcolor{violet}{$\circ$}'' shows \textcolor{violet}{\textsf{ACORN}}, and ``\textcolor{red}{$+$}'' shows \textcolor{red}{\textsf{Grant-pnv}}.}
        \label{fig:k}
    \end{center}
\end{figure*}
\begin{figure*}[!t]
    \begin{center}
        \subfigure[Deep]{%
            \includegraphics[width=0.23\linewidth]{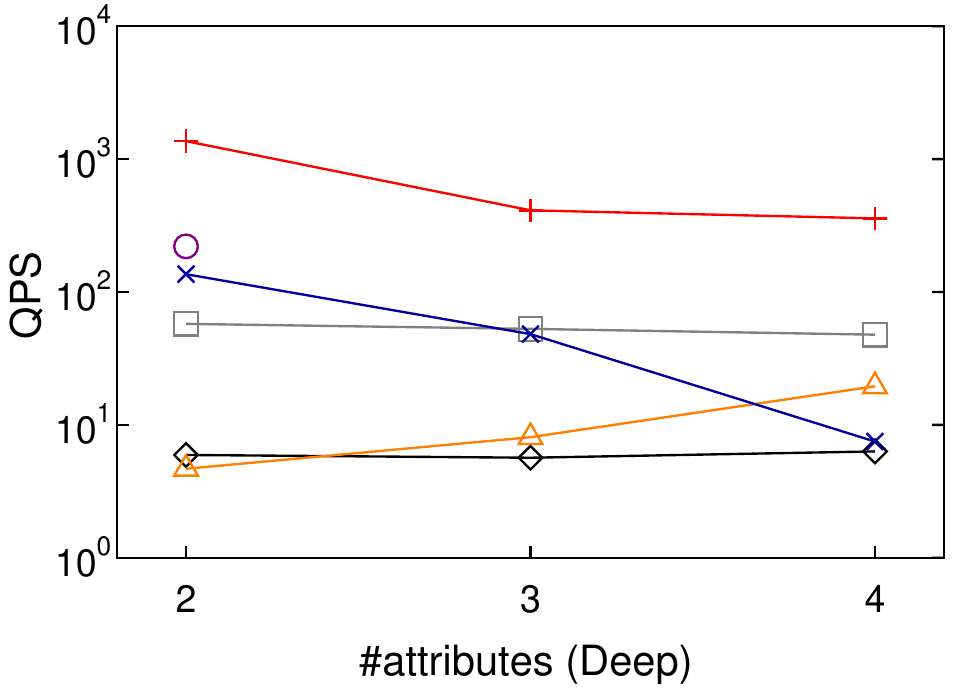}    \label{fig:deep_m}}
        \subfigure[SIFT]{%
            \includegraphics[width=0.23\linewidth]{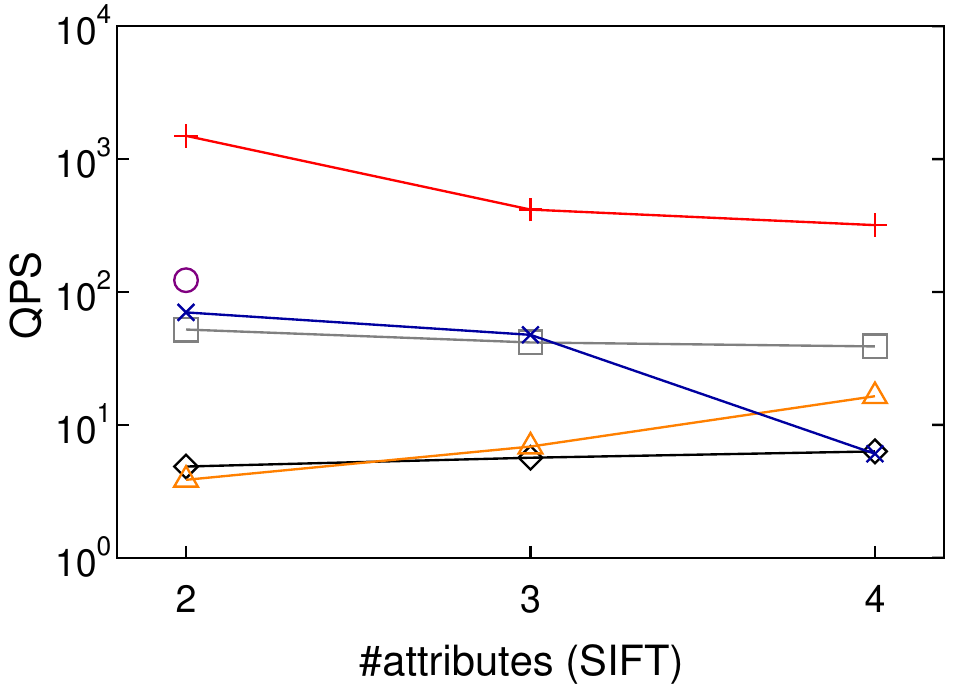}    \label{fig:sift_m}}
        \subfigure[SPACEV]{%
            \includegraphics[width=0.23\linewidth]{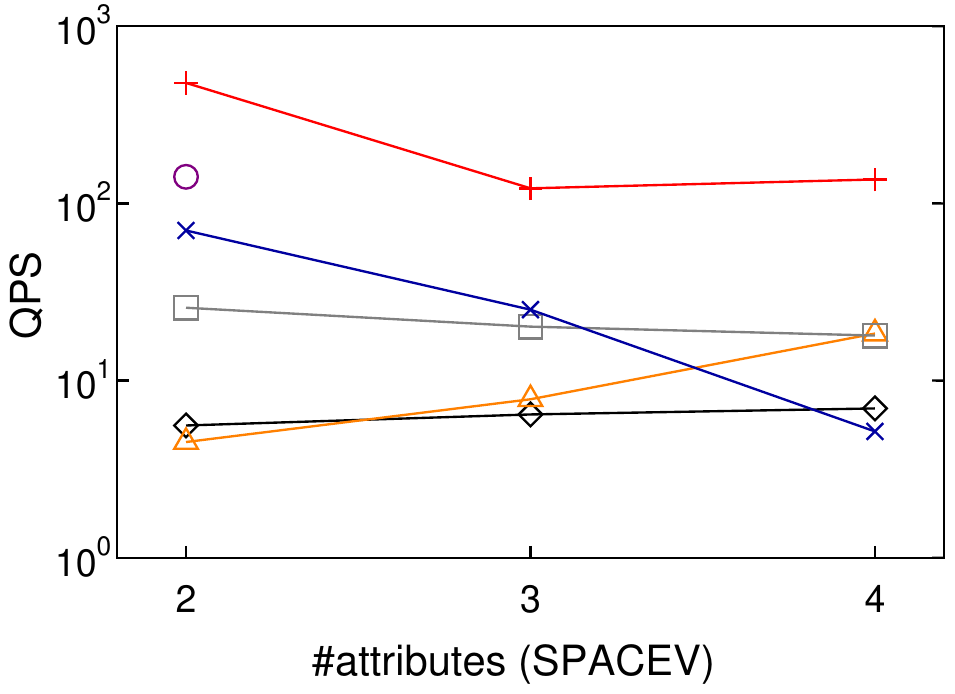}  \label{fig:spacev_m}}
        \subfigure[T2I]{%
            \includegraphics[width=0.23\linewidth]{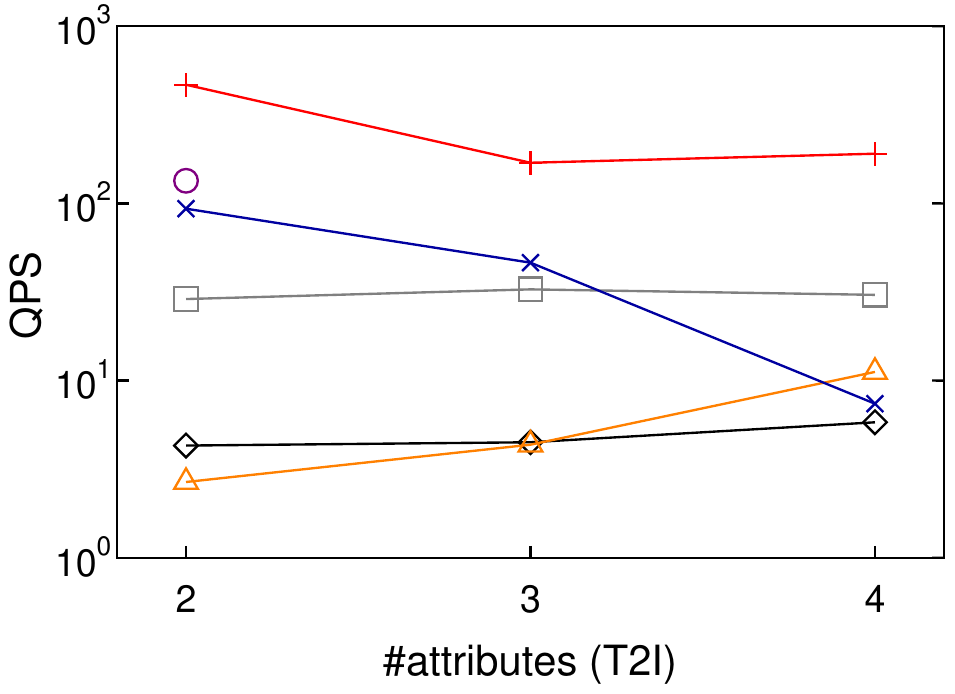}     \label{fig:t2i_m}}
        \caption{Impact of the number of attributes under recall $\approx 0.9$.
            ``$\diamond$'' shows \textsf{BF}, ``\textcolor{orange}{$\triangle$}'' shows \textcolor{orange}{\textsf{Pre-filter}}, ``\textcolor{gray}{$\square$}'' shows \textcolor{gray}{\textsf{IVF}}, ``\textcolor{blue}{$\times$}'' shows \textcolor{blue}{\textsf{Post-filter}}, ``\textcolor{violet}{$\circ$}'' shows \textcolor{violet}{\textsf{ACORN}}, and ``\textcolor{red}{$+$}'' shows \textcolor{red}{\textsf{Grant-pnv}}.}
        \label{fig:m}
    \end{center}
\end{figure*}
\begin{figure*}[!t]
    \begin{center}
        \subfigure[Deep]{%
            \includegraphics[width=0.23\linewidth]{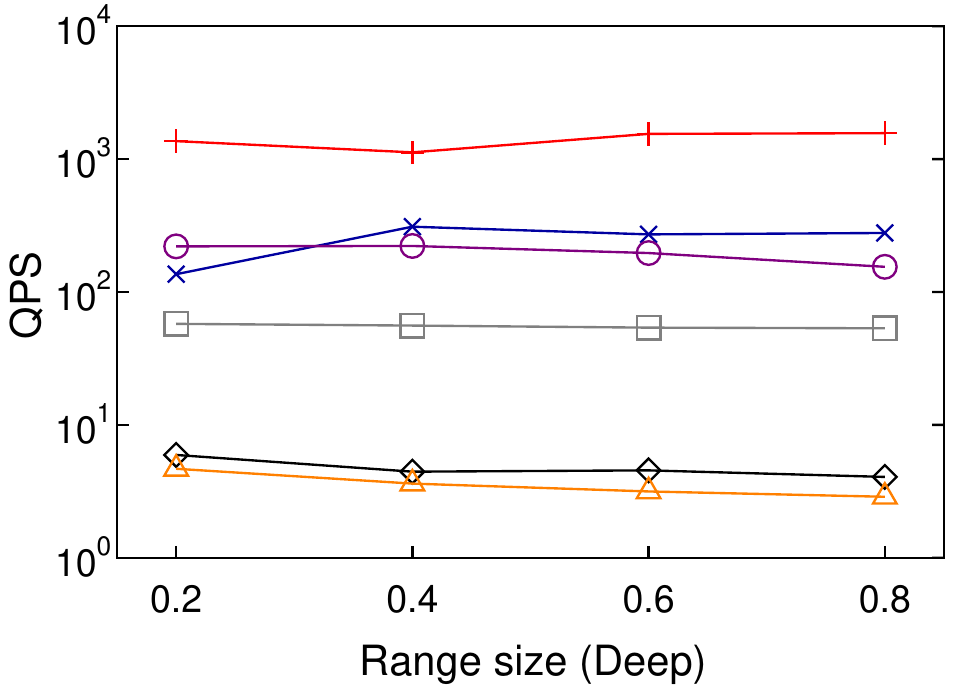}    \label{fig:deep_range}}
        \subfigure[SIFT]{%
            \includegraphics[width=0.23\linewidth]{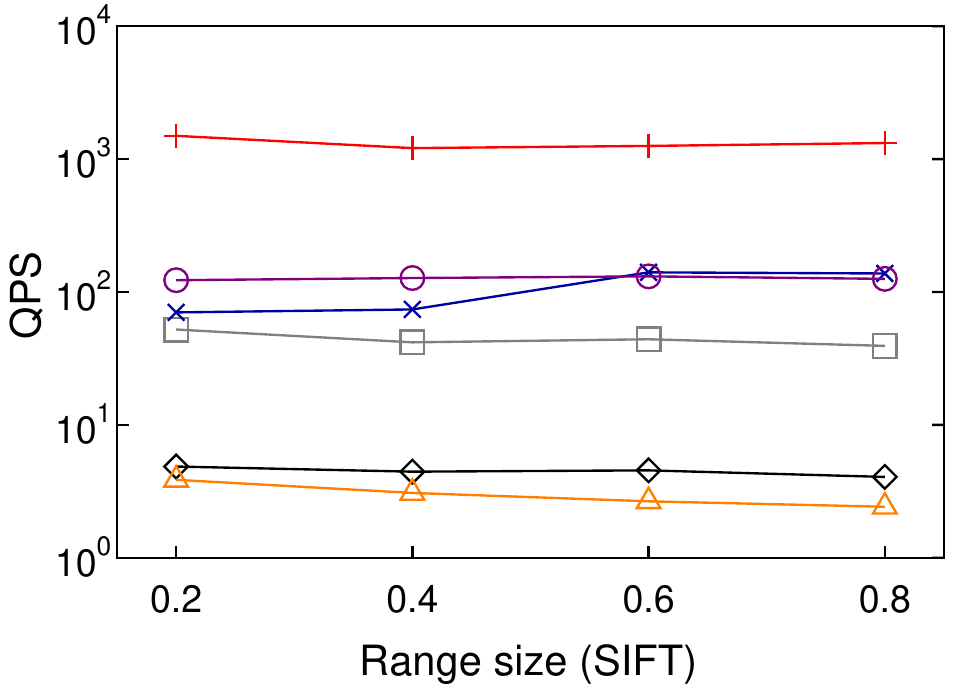}    \label{fig:sift_range}}
        \subfigure[SPACEV]{%
            \includegraphics[width=0.23\linewidth]{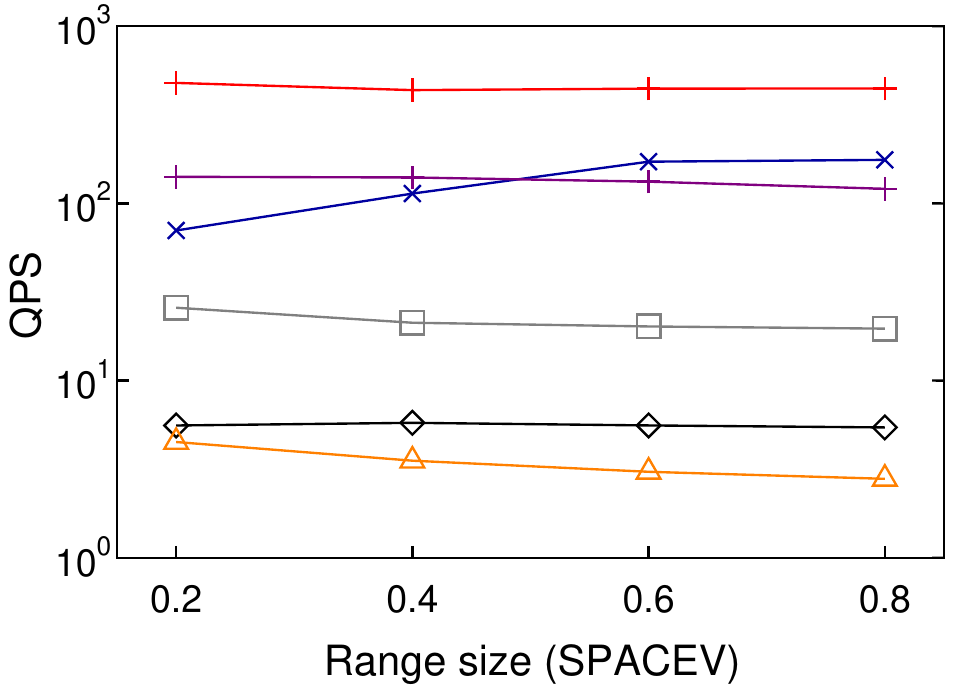}  \label{fig:spacev_range}}
        \subfigure[T2I]{%
            \includegraphics[width=0.23\linewidth]{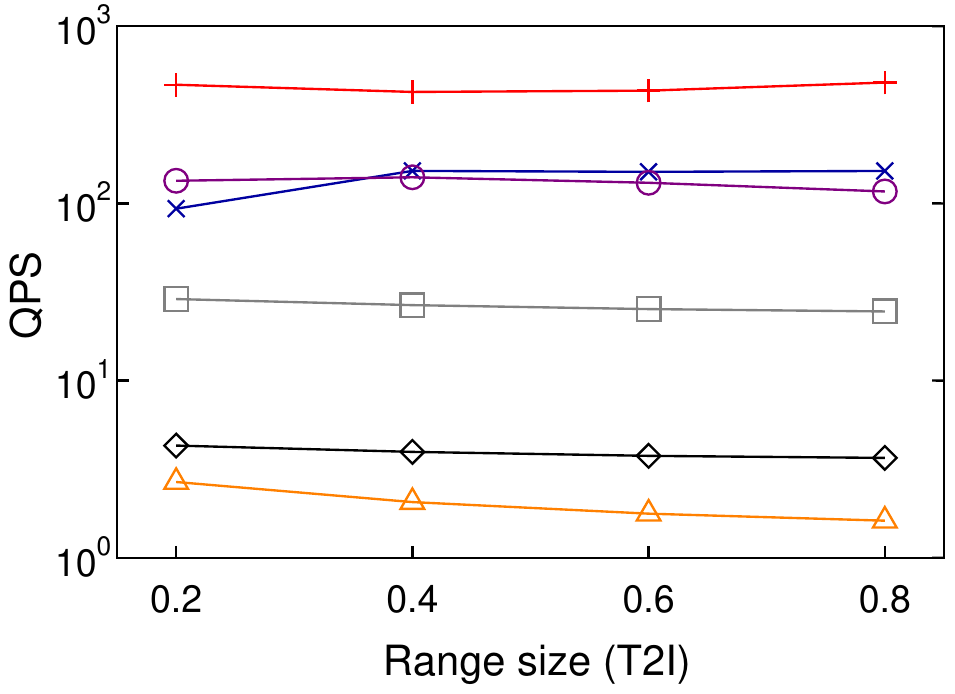}     \label{fig:t2i_range}}
        \caption{Impact of range size under recall $\approx 0.9$.
            ``$\diamond$'' shows \textsf{BF}, ``\textcolor{orange}{$\triangle$}'' shows \textcolor{orange}{\textsf{Pre-filter}}, ``\textcolor{gray}{$\square$}'' shows \textcolor{gray}{\textsf{IVF}}, ``\textcolor{blue}{$\times$}'' shows \textcolor{blue}{\textsf{Post-filter}}, ``\textcolor{violet}{$\circ$}'' shows \textcolor{violet}{\textsf{ACORN}}, and ``\textcolor{red}{$+$}'' shows \textcolor{red}{\textsf{Grant-pnv}}.}
        \label{fig:range}
    \end{center}
\end{figure*}

\bibliography{acmart}
\bibliographystyle{abbrvnat}

\end{document}